\documentclass{aims}

\usepackage{txfonts,cite}
\usepackage[all]{xy}
\usepackage{algorithm,algpseudocode} 
\usepackage{algorithmicx}
\usepackage{graphicx}
\usepackage{subcaption}
\def\typeofarticle{Type of article}
\def\currentvolume{x}
\def\currentissue{x}
\def\currentyear{2024}

\def\ppages{xxx--xxx}
\def\DOI{}
\def\Received{}
\def\Revised{}
\def\Accepted{}
\def\Published{}

\numberwithin{equation}{section}

\newcommand\scalemath[2]{\scalebox{#1}{\mbox{\ensuremath{\displaystyle #2}}}}

\def\bfb{{\boldsymbol{b}}}
\def\bfc{{\mathbf{c}}}

\def\bfn{{\boldsymbol{n}}}

\def\bfx{{\boldsymbol{x}}}
\def\bfy{{\boldsymbol{y}}}
\def\bfu{{\boldsymbol{u}}}
\def\bfv{{\boldsymbol{v}}}

\def\bfN{{\boldsymbol{N}}}

\newtheorem{theorem}{Theorem}
\newtheorem{corollary}[theorem]{Corollary}
\newtheorem{proposition}[theorem]{Proposition}
\newtheorem{lemma}[theorem]{Lemma}
\theoremstyle{definition}

\newtheorem{example}{Example}
\newtheorem{remark}{Remark}

\begin{document}

\title{Computation of symmetries of rational surfaces}

\author{%
  Juan Juan Gerardo Alc\'azar\affil{1,}\corrauth,
  Carlos Hermoso\affil{1}, 
  H\"usn\"u An{\i}l \c{C}oban\affil{2}
  and
  U\u{g}ur G\"oz\"utok\affil{3}
}

\shortauthors{the Author(s)}

\address{%
  \addr{\affilnum{1}}{Departamento de F\'{ı}sica y Matem\'{a}ticas, Universidad de Alcal\'{a}, E-28871 Madrid, Spain}
  \addr{\affilnum{2}}{Department of Mathematics, Karadeniz Technical University, Trabzon, T\"urkiye}
  \addr{\affilnum{3}}{Department of Natural Sciences, National Defence University, Istanbul, T\"urkiye}}

\corraddr{Email: juange.alcazar@uah.es.}

\begin{abstract}
In this paper we provide, first, a general symbolic algorithm for computing the symmetries of a given rational surface, based on the classical differential invariants of surfaces, i.e. {\it Gauss curvature} and {\it mean curvature}. In practice, the algorithm works well for sparse parametrizations (e.g. toric surfaces) and PN surfaces. Additionally, we provide a specific, also symbolic algorithm for computing the symmetries of ruled surfaces; this algorithm works extremely well in practice, since the problem is reduced to that of rational space curves, which can be efficiently solved by using existing \textcolor{black}{methods}. The algorithm for ruled surfaces is based on the fact, proven in the paper, that every symmetry of a rational surface must also be a symmetry of its {\it line of striction}, which is a rational space curve. The algorithms have been implemented in the computer algebra system {\tt Maple}, and the implementations have been made public; evidence of their performance is given in the paper.  
\end{abstract}


\maketitle

\section{Introduction}

\textcolor{black}{The topic addressed in this paper is the computation of the Euclidean symmetries of a certain type of algebraic objects, namely rational surfaces. Thus, in the rest of this paper, a ``symmetry" will mean an isometry of the ambient space that leaves the embedded surface invariant; in particular, such a mapping also leaves invariant the Gauss and mean curvatures of the surface. Capturing} the information on the Euclidean symmetries of an object is a natural step if one wants to understand its geometry, visualize it accurately or store it efficiently. Although in Applied Mathematics the computation of symmetries has been addressed for decades, with a very wide variety of techniques often involving Numerics and Statistical Methods, in the past fifteen years there has been a growing interest in applying methods from Symbolic Computation to this problem. In this sense, in Applied Mathematics the object to be analyzed is usually a 2D or 3D image, and in practice often a 2D or 3D point cloud, and one seeks {\it approximate} symmetries, rather than {\it exact} ones. When the problem is seen from the \textcolor{black}{point of view} of Symbolic Computation, which is the perspective we use in this paper, one is attracted towards exact objects with more structure, and looks for their exact symmetries. Typically, the objects considered are (algebraic) varieties defined by polynomials with integer coefficients, either implicitly given, i.e. as zeros of a given set of polynomials, or defined by means of a parametrization with rational components, i.e. quotients of polynomials.

Although implicitly defined varieties have been addressed (see \cite{ALV,Leb1,Leb2} for implicit planar curves, and \cite{ALV23b,BLV,GH23} for implicit surfaces), the case that has attracted more attention has been that of rational curves and surfaces, probably because the rational representation is preferred in Computer Aided Geometric Design, a field where the problem has some relevance. For rational parametric curves the problem is \textcolor{black}{quite} well-understood, and in fact a more general question, namely identifying whether or not two rational curves are affinely or projectively equivalent, \textcolor{black}{has been investigated in many papers}. Although a complete list would be long, some notable contributions are \cite{AHM15}, on symmetries of space curves, and \cite{BLV,HJ18}, on projective equivalences of rational curves in any dimension; see the bibliographies of these papers for other references. \textcolor{black}{Observe also that isometries, affine and projective equivalences are special instances of projective equivalences.}

However, for rational surfaces the problem is harder, both theoretically and computationally. Here we have two types of contributions, a first group providing more general algorithms, and a second group focused on special types of surfaces whose properties one can exploit. Among the first group, \cite{AH16} considers involutions of parametric polynomial surfaces, \cite{HJS19} addresses affine and projective equivalences of rational surfaces without projective base points (although base points, however, are not uncommon\textcolor{black}{; note that rational curves do not have base points, which makes the problem easier for curves}), and \cite{JLS22}, which improves on the results of \cite{HJS19} and provides an algorithm for computing the projective equivalences between two rational surfaces without additional requirements. The algorithm in \cite{JLS22} is certainly very general, but has still some drawbacks: on one hand, the problem is reduced to five possible cases, but only two of them are solved in the paper; on the other hand, \textcolor{black}{the algorithm is involved (check the partial implementation in \cite{github}) and requires a strong background of Algebraic Geometry, so it is not so easily accessible for readers not familiar with such topic. Let us also mention that in \cite{HJS19} the input is a polynomial or rational parametrization of a surface, with real coefficients, in any dimension; in \cite{JLS22} the input is a rational parametrization of a surface with coefficients over an algebraically closed field.}

Among the second group of papers, namely papers related to rational surfaces with special properties, one has \cite{ADM18}, for canal surfaces, \cite{AM22}, for translational and minimal surfaces, and \cite{AQ,BLV} for ruled surfaces; in fact, \cite{AM22,AQ} address the more general question of computing affine equivalences, and \cite{BLV} studies projective equivalences.  

In this paper we present two contributions to the problem. The first one is a general algorithm which uses \textcolor{black}{the} two natural invariants for surfaces, namely Gauss curvature and mean curvature, therefore somehow generalizing to surfaces the idea of \cite{AHM15} for space curves. The algorithm is very easy to implement, although it can be computationally costly. The reason is that the algorithm requires to compute the resultant of two polynomials derived from the two curvatures; however, while the Gauss curvature of a rational surface is always rational, the mean curvature in general has a square-root, so we need to use the square of the mean curvature instead, which increases the degree. If we work with a {\it PN surface}, namely a rational surface where the modulus of the normal vector is rational (see for instance \cite{KMV}), we can skip that part and reduce the computational cost, so the algorithm works \textcolor{black}{better}; otherwise, we need very sparse parametrizations. As for a comparison with existing algorithms, compared to \cite{HJS19} projective base points are not a problem for our algorithm\textcolor{black}{, although unlike ours, \cite{HJS19} also works in higher dimension}; compared to \cite{JLS22}, our algorithm is much simpler. However, it is true that for non-sparse, non-PN surfaces, the computational cost can be high\textcolor{black}{, so in the cases covered by \cite{HJS19,JLS22} and without further advantages that make our algorithm competitive (e.g. the cases already mentioned or ruled surfaces, see the next paragraph), using \cite{HJS19,JLS22} could be preferrable}. 

The second contribution is a specific algorithm for rational ruled surfaces. We prove that every symmetry of a rational ruled surface is a symmetry of the {\it line of striction} of the surface, which is a rational space curve. Thus, we reduce the problem to that of space curves, which can be efficiently solved. The resulting algorithm provides better timings than the general algorithm described in the previous paragraph and the algorithm in \cite{AQ}, as we show with several examples. Compared to \cite{BLV}, it is worth noticing that the algorithm in \cite{BLV} also reduces the problem to the case of curves. However, there are two main differences with our algorithm: (i) \cite{BLV} uses rational curves in ${\Bbb P}^5$, so the dimension of the curves is higher; (ii) the algorithm in \cite{BLV} requires to solve a quadratic system of equations in 7 parameters, which may cause difficulties in some cases (see page 8 of \cite{BLV}). It is true though that the algorithm in \cite{BLV} can solve the more general problem of finding the projective equivalences between two ruled surfaces, while our algorithm only aims to computing the symmetries of a given surface.  

The two algorithms presented in the paper have been implemented in Maple, and are available in the website of one of the authors \cite{website}\textcolor{black}{; a \emph{Zenodo} version is also available \cite{zenodo}}. 

The structure of the paper is the following. In Section \ref{background} we provide the necessary background to develop our results. The general algorithm is presented in Section \ref{gen-alg}. The algorithm for ruled surfaces is given in Section \ref{sec-specruled}, where we also analyze the special case of developable surfaces. Section \ref{sec-experimentation} reports on the experimentation carried out on different types of surfaces, general and ruled, using the implementation carried out in Maple. We close with our conclusion in Section \ref{conclusion}. 

\section{General background, notation and hypotheses}\label{background}

Let $\bfx(t,s)$ be a rational parametrization of a surface $S\subset {\Bbb R}^3$. We say that $\bfx(t,s)$ is a \emph{proper} parametrization of $S$ if the parametrization $\bfx$, seen as a rational mapping $\bfx:{\Bbb R}^2\to {\Bbb R}^3$, is birational \textcolor{black}{onto its image}, so that the inverse $\bfx^{-1}$ exists and is also rational. Thus, proper mappings are injective except perhaps at a subset of $S$ of dimension at most 1. Checking whether or not a surface parametrization is proper, and reparametrizing it to be proper in the case \textcolor{black}{when} it is not, is not easy, but one can always heuristically check properness by taking a random point $(t_0,s_0)$, and solving $\bfx(t,s)=\bfx(t_0,s_0)$: if the only solution is $t=t_0,s=s_0$, then $\bfx(t,s)$ is proper with high probability. 

\par
The notion of properness also applies to rational parametrizations $\bfc(t)$ of a curve ${\mathcal C}\subset {\Bbb R}^3$. In this case, however, checking properness and \textcolor{black}{reparametrizing} non-proper parametrizations is much easier: once that all the components of $\bfc(t)$ have the same denominator, checking properness amounts to computing the gcd of the components of the numerators of $\bfc(t)-\bfc(s)$. If the degree of this gcd is 1, then $\bfc(t)$ is proper; if it is not, one can compute a proper reparametrization using the information provided by this gcd. The interested reader can read further on this topic in \cite{SWP}.

\par
A \emph{Euclidean symmetry} (also called \emph{isometry}), of ${\Bbb R}^3$ is a mapping $f:{\Bbb R}^3\to {\Bbb R}^3$, $f(\bfx)=A\bfx+\bfb$, where $A$ is an orthogonal $3\times 3$ matrix, so $A^TA=I$ with $I$ the identity matrix, and $\bfb\in {\Bbb R}^3$; notice that $\mbox{det}(A)=\pm 1$. We say that such a mapping $f$ is a \emph{symmetry} of a surface $S$ (resp. a curve ${\mathcal C}$) if $f(S)=S$ (resp. $f({\mathcal C})={\mathcal C}$). When we work with proper, rational parametrizations of curves or surfaces, one can reduce the computation of the symmetries of the curve or surface, if any, to the existence of birational parametrizations in the parameter space, ${\Bbb R}$ in the case of curves and ${\Bbb R}^2$ in the case of surfaces, satisfying a certain condition; see the \textcolor{black}{result} that follows. For curves, we can go a bit further taking into account that the birational mapping\textcolor{black}{s} of the real line \textcolor{black}{(and in fact of the projective complex line)}  are the \emph{M\"obius transformations}, i.e. the mappings
\[
\varphi(t)=\dfrac{at+b}{ct+d},\mbox{ }ad-bc\neq 0.
\]
More precisely, we have the following \textcolor{black}{result}; see for instance \cite{AHM15} for a proof.

\begin{proposition} \label{ch7-fund-th-curves}
Let ${\mathcal C}\subset {\Bbb R}^3$ be a rational curve defined by a rational, proper parametrization $\bfc (t)$, and let $f:{\Bbb R}^3\to {\Bbb R}^3$ be a symmetry of ${\mathcal C}$. Then there exists a M\"obius transformation $\varphi:{\Bbb R}\to {\Bbb R}$ such that 
\begin{equation}\label{ch7-fund-curves}
f\circ \bfx=\bfx\circ \varphi,
\end{equation}
i.e. making commutative the following diagram:
\begin{equation}\label{ch7-eq:fundamentaldiagram-curves}
\xymatrix{
{\mathcal C} \ar[r]^{f} & {\mathcal C} \\
{\Bbb R} \ar@{-->}[u]^{\bfc} \ar@{-->}[r]_{\varphi} & {\Bbb R} \ar@{-->}[u]_{\bfc}
}.
\end{equation}
\end{proposition}

For \textcolor{black}{properly parametrized} surfaces we \textcolor{black}{also have a commutative diagram analogous to Eq. \eqref{ch7-eq:fundamentaldiagram-curves}, an observation used} in papers like \cite{AH16,HJS19}. Indeed, if $S\subset {\Bbb R}^3$ \textcolor{black}{is} a rational surface defined by a rational, proper parametrization $\bfx (t,s)$, and $f:{\Bbb R}^3\to {\Bbb R}^3$ \textcolor{black}{is} a symmetry of $S$\textcolor{black}{, then there exists a birational mapping $\psi:{\Bbb R}^2\to {\Bbb R}^2$ making commutative the diagram}
\begin{equation}\label{ch7-eq:fundamentaldiagram}
\xymatrix{
S \ar[r]^{f} & S \\
{\Bbb R}^2 \ar@{-->}[u]^{\bfx} \ar@{-->}[r]_{\psi} & {\Bbb R}^2 \ar@{-->}[u]_{\bfx}
},
\end{equation}
\textcolor{black}{i.e.} such that 
\begin{equation}\label{ch7-fund}
f\circ \bfx=\bfx\circ \psi.
\end{equation}

\begin{color}{black}
\begin{remark}
In fact Proposition \ref{ch7-fund-th-curves} holds whenever we have a birational transformation, not necessarily a Euclidean symmetry, $f:{\Bbb R}^n\to {\Bbb R}^n$ mapping two rational curves ${\mathcal C}_1,{\mathcal C}_2$ properly parametrized by $\bfx_1,\bfx_2$. Similarly, the diagram in Eq. \eqref{ch7-eq:fundamentaldiagram} works also when we have two surfaces $S_1,S_2$, properly parametrized, and a birational transformation mapping one of them onto the other. The hypothesis on the properness of the parametrizations is key in both cases. 
\end{remark} 
\end{color}

Birational mappings of the \textcolor{black}{(projective)} plane are called \emph{Cremona transformations}, and are known to be generated by linear projective mappings and quadratic transformations. However, while the birational mappings of the \textcolor{black}{complex and real} line are, exactly, the M\"obius transformations, and therefore we have a closed form for them, no closed form is known for the Cremona transformations. This is the reason why in several papers related to the problem treated here, e.g. \cite{AH16,AQ,HJS19}, one studies rational parametrizations with additional properties so that a general form for the Cremona transformations behind $f$ can be guessed. In \cite{AH16}, under the considered hypotheses, the transformation is linear; in \cite{HJS19}, linear projective; in \cite{AQ}, the first component of the transformation is a univariate M\"obius mapping, while the second component is more complicated, but with a known general form. \textcolor{black}{Notice that, in general, M\"obius transformations or Cremona transformations do not preserve the metric.}

\subsection{Background on Differential Geometry}

Let us briefly recall some notions from Differential Geometry that will be needed later; we refer to \cite{DoCarmo, Gray99, Struik} for more information on this topic. 

\par
Let $\bfx(t,s)$ be a parametrization of a surface $S\subset {\Bbb R}^3$. We say that $\bfx(t,s)$ is \emph{regular} at $(t_0,s_0)\in {\Bbb R}^2$ if the partial derivatives $\bfx_t,\bfx_s$ are linearly independent at the point. At regular points $P=\bfx(t,s)$ the \emph{tangent plane} to $S$ is well-defined and corresponds to 
\[
T_P(S)=P+{\mathcal L}(\bfx_t,\bfx_s),
\]
where ${\mathcal L}(\bfx_t,\bfx_s)$ denotes the linear variety spanned by $\bfx_t,\bfx_s$. The normal direction to the tangent plane is therefore that of 
\begin{equation}\label{firstnormal}
\bfN_\bfx(t,s)=\bfx_t\times \bfx_s.
\end{equation}
Additionally, the unitary \emph{normal vector} is defined as
\begin{equation}\label{ch7-normal}
\bfn_\bfx(t,s)=\dfrac{\bfx_t\times \bfx_s}{\Vert \bfx_t\times \bfx_s\Vert}.
\end{equation}
Notice that if $\bfx(t,s)$ is rational, in general $\bfn_\bfx(t,s)$ is not rational, since a square-root is expected in the denominator. Additionally, since $\bfn_\bfx(t,s)$ is unitary one can also see $\bfn_\bfx$ as a mapping from the surface $S$ onto the unit sphere ${\Bbb S}^2$, i.e. $\bfn_\bfx:S\to {\Bbb S}^2$; this mapping is called the \emph{Gauss map}. 

\par
The \emph{first fundamental form} and the \emph{second fundamental form} are two symmetric quadratic forms defined at each tangent plane $T_P(S)$ of $S$, with $P$ regular. The first fundamental form captures the metric of the surface, and is defined by the matrix
\begin{equation}\label{ch7-firstfundam}
{\bf I}_\bfx=\begin{bmatrix}
E & F \\
F & G
\end{bmatrix}=\begin{bmatrix}
\bfx_t\cdot \bfx_t & \bfx_t\cdot \bfx_s \\
\bfx_t\cdot \bfx_s & \bfx_s\cdot \bfx_s
\end{bmatrix}.
\end{equation}
The second fundamental form captures the local shape of the surface around the point, and is defined by the matrix
\begin{equation}\label{ch7-secondfundam}
{\bf II}_\bfx=\begin{bmatrix}
L & M \\
M & N
\end{bmatrix}=\begin{bmatrix}
\bfx_{tt}\cdot \bfn_\bfx & \bfx_{ts}\cdot \bfn_\bfx \\
\bfx_{ts}\cdot \bfn_\bfx & \bfx_{ss}\cdot \bfn_\bfx
\end{bmatrix}.
\end{equation}

The \emph{Weingarten map}\index{Weingarten map} ${\mathcal W}_\bfx$ is the (linear) mapping ${\mathcal W}_\bfx:T_P(S)\to T_P(S)$ defined as ${\mathcal W}_\bfx=-d\bfn_\bfx$, i.e. the differential of the Gauss map with a changed sign. The matrix defining ${\mathcal W}_\bfx$ is 
\begin{equation}\label{ch7-weingarten}
{\bf W}_\bfx={\bf I}_\bfx^{-1} {\bf II}_\bfx.
\end{equation}
The determinant of ${\bf W}_\bfx$ is the \emph{Gauss curvature}\index{Gauss curvature} of the surface, ${\bf K}_\bfx$, and the trace of ${\bf W}_\bfx$ is the \emph{mean curvature}\index{mean curvature}, ${\bf H}_\bfx$. In terms of the entries of the matrices ${\bf I}_\bfx,{\bf II}_\bfx$ (see for instance Section 3.4 of \cite{PMaek}), 
\begin{equation}\label{ch7-gausswithcomp}
{\bf K}_\bfx=\dfrac{LN-M^2}{EG-F^2},\mbox{ }{\bf H}_\bfx=\dfrac{EN+GL-2FM}{2(EG-F^2)}.
\end{equation}
Observe that if $\bfx(t,s)$ is rational, ${\bf K}_\bfx$ is always a rational function. However, in general ${\bf H}_\bfx$ is not rational because of the presence of a square-root \textcolor{black}{(due to the functions $L,M,N$ in the numerator of the formula for ${\bf H}_\bfx$)}, but ${\bf H}_\bfx^2$ is certainly rational.

\section{A general algorithm} \label{gen-alg}

Let $\bfx(t,s)$ be a rational, proper parametrization of a surface $S\subset {\Bbb R}^3$. Our goal is to provide an algorithm for computing the symmetries $f$ of $S$, if any. In order to compute the symmetries $f$, the idea, from \textcolor{black}{the diagram in Eq. \eqref{ch7-eq:fundamentaldiagram}}, is to first find the Cremona transformations $\psi$ satisfying $f\circ \bfx=\bfx\circ \psi$, and then compute the $f$ themselves. Notice that if $f(\bfx)=A\bfx+\bfb$ is a symmetry of $S$, Eq. \eqref{ch7-fund} yields $A\bfx(t,s)+\bfb=(\bfx\circ \psi)(t,s)$; if $\psi$ is known, one can compute $A,\bfb$ by directly solving the linear system stemming from this equality, or by solving a linear system computed from $A\bfx(t,s)+\bfb=(\bfx\circ \psi)(t,s)$ after giving enough values $(t_i,s_i)$. 

\indent 
In order to find $\psi$ we will use the Gauss curvature ${\bf K}_\bfx$ and the mean curvature ${\bf H}_\bfx$ of the parametrization. The first step is to find the relationships between: (i) ${\bf K}_\bfx, {\bf H}_\bfx$ and ${\bf K}_{f\circ \bfx},{\bf H}_{f\circ \bfx}$; (ii) ${\bf K}_\bfx, {\bf H}_\bfx$ and ${\bf K}_{\bfx\circ \psi},{\bf H}_{\bfx\circ \psi}$. In turn, this requires to see how $\bfn_\bfx,{\bf I}_\bfx,{\bf II}_\bfx$ are affected when moving from $\bfx$ to $f\circ \bfx$ or $\bfx\circ \psi$. We start with $f\circ \bfx$. 

\begin{lemma}\label{ch7-A-transform}
Let $\bfx(t,s)$ be a rational parametrization of a surface $S$, and let $f(\bfx)=A\bfx+\bfb$ be a Euclidean symmetry. 
\begin{itemize}
\item [(1)] ${\bf I}_{f\circ \bfx}={\bf I}_{\bfx}$.
\item [(2)] $\bfn_{f\circ \bfx}=\det(A) A \bfn_{\bfx}$. 
\item [(3)] ${\bf II}_{f\circ \bfx}=\det(A) {\bf II}_{\bfx}$. 
\end{itemize}
\end{lemma}

\begin{proof} To see (1), let $\tilde{\bfx}=f\circ \bfx$, so $\tilde{\bfx}(t,s)=A\bfx(t,s)+\bfb$. Then $\tilde{\bfx}_t=A\bfx_t$, $\tilde{\bfx}_s=A\bfx_s$. Since $A$ is orthogonal, $\tilde{\bfx}_t\cdot \tilde{\bfx}_t=\bfx_t\cdot \bfx_t$, $\tilde{\bfx}_t\cdot \tilde{\bfx}_s=\bfx_t\cdot \bfx_s$, $\tilde{\bfx}_s\cdot \tilde{\bfx}_s=\bfx_t\cdot \bfx_t$, and (1) follows. To see (2), we observe that 
\[
\tilde{\bfx}_t\times \tilde{\bfx}_s=(A\bfx_t)\times (A\bfx_s)=\mbox{det}(A) A^{-T}(\bfx_t\times \bfx_s).
\]
Since $A$ is orthogonal, $A^{-T}=A$, and therefore $\tilde{\bfx}_t\times \tilde{\bfx}_s=\mbox{det}(A) A(\bfx_t\times \bfx_s)$. Then (2) follows by using Eq. \eqref{ch7-normal}, taking into account that since $A$ is orthogonal, norms are preserved. Finally, for (3) we observe first that $\tilde{\bfx}_{tt}=A\bfx_{tt}$, $\tilde{\bfx}_{ts}=A\bfx_{ts}$, $\tilde{\bfx}_{ss}=A\bfx_{ss}$; then the result follows from the definition of the second fundamental form and statement (2), again taking into account that since $A$ is orthogonal, dot products are preserved. 
\end{proof}

\begin{corollary}\label{ch7-cor-A-transform}
Let $\bfx(t,s)$ be a rational parametrization of a surface $S$, and let $f(\bfx)=A\bfx+\bfb$ be a Euclidean symmetry. Then 
\begin{equation}\label{ch7-wein-gauss-mean}
{\bf W}_{f\circ \bfx}=\det(A){\bf W}_\bfx,\mbox{ }{\bf K}_{f\circ \bfx}={\bf K}_\bfx,\mbox{ }{\bf H}_{f\circ \bfx}=\det(A) {\bf H}_\bfx.
\end{equation}
\end{corollary}

\begin{proof} The equality ${\bf W}_{f\circ \bfx}=\det(A){\bf W}_\bfx$ follows from Eq. \eqref{ch7-weingarten}, taking into account the statements (1) and (3) of Lemma \ref{ch7-A-transform}. The remaining equalities follow from the fact that ${\bf K}_{f\circ \bfx}$ and ${\bf H}_{f\circ \bfx}$ are the determinant and trace of the matrix ${\bf W}_{f\circ \bfx}$.
\end{proof} 

\par
Now let us analyze the relationship between ${\bf K}_\bfx, {\bf H}_\bfx$ and ${\bf K}_{\bfx\circ \psi},{\bf H}_{\bfx\circ \psi}$. In order to do this we need to recall first the relationship between the first and second fundamental forms of $\bfx$ and $\bfx\circ \psi$; this is essentially known, but we include \textcolor{black}{it} here for the convenience of the reader. Let us denote $\bfy:=\bfx\circ \psi$, $\psi(t,s)=(\psi_1(t,s),\psi_2(t,s))$. Using the Chain rule, 
\begin{equation}\label{ch7-chainrule}
\begin{array}{rcl}
\bfy_t &=& \bfx_t(\psi)\frac{\partial \psi_1}{\partial t}+\bfx_s(\psi)\frac{\partial \psi_2}{\partial t},\\
\bfy_s &=& \bfx_t(\psi)\frac{\partial \psi_1}{\partial s}+\bfx_s(\psi)\frac{\partial \psi_2}{\partial s},
\end{array}
\end{equation}
where $\bfy_t,\bfy_s$ are evaluated at $(t,s)$, and $\bfx_t(\psi),\bfx_s(\psi)$ denote the evaluations of $\bfx_t,\bfx_s$ at $\psi(t,s)$. In matrix notation, we have 
\[
\begin{bmatrix}
\bfy_t\\ \bfy_s 
\end{bmatrix}= J^T\cdot \begin{bmatrix} \bfx_t(\psi)\\ \bfx_s(\psi) \end{bmatrix},
\]
where $J$ represents the Jacobian of $\psi$. Taking Eq. \eqref{ch7-firstfundam} into account, we get that 
\[
{\bf I}_{\bfy}=J^T \cdot {\bf I}_\bfx(\psi)\cdot J,
\]
where, again, the entries of the first fundamental form at the left-hand side must be evaluated at $(t,s)$, while at the right-hand side the entries of the form are evaluated at $\psi(t,s)$. To relate $\bfn_\bfy$ and $\bfn_\bfx$ we also use Eq. \eqref{ch7-chainrule}. Finally, to relate ${\bf II}_\bfy$ and ${\bf II}_\bfx$ we need to differentiate Eq. \eqref{ch7-wein-gauss-mean} again, applying the Chain rule; although the calculations are lengthy, we get the same result that we obtained for the first fundamental form. So we have the following lemma. 

\begin{lemma}\label{ch7-A-transform2}
Let $\bfx(t,s)$ be a rational parametrization of a surface $S$, let $\psi$ be a planar rational mapping, and let $J$ be the Jacobian of $\psi$.
\begin{itemize}
\item [(1)] ${\bf I}_{\bfx\circ \psi}=J^T \cdot{\bf I}_\bfx(\psi) \cdot J$.
\item [(2)] $\bfn_{\bfx\circ \psi}=\frac{\det(J)}{|\det(J)|}\bfn_{\bfx}(\psi)$. 
\item [(3)] ${\bf II}_{\bfx\circ \psi}=J^T\cdot {\bf II}_\bfx(\psi) \cdot J$.
\end{itemize}
\end{lemma}

Now we can prove what we need. 

\begin{corollary}\label{ch7-cor-A-transform2}
Let $\bfx(t,s)$ be a rational parametrization of a surface $S$, let $\psi$ be a planar birational mapping, and let $J$ be the Jacobian of $\psi$. Then
\begin{equation}\label{ch7-wein-gauss-mean2}
{\bf W}_{\bfx\circ \psi}=J^{-1}\cdot {\bf W}_\bfx(\psi)\cdot J,
\end{equation}
and 
\begin{equation}\label{ch7-wein-gauss-mean3}
{\bf K}_{\bfx\circ \psi}={\bf K}_\bfx(\psi),\mbox{ }{\bf H}_{\bfx\circ \psi}={\bf H}_\bfx(\psi).
\end{equation}
\end{corollary}

\begin{proof} From Eq. \eqref{ch7-weingarten} and statements (1), (3) of Lemma \ref{ch7-A-transform2},
\[
{\bf W}_{\bfx\circ \psi}=\left(J^T \cdot{\bf I}_\bfx(\psi) \cdot J\right)^{-1}\left(J^T\cdot {\bf II}_\bfx(\psi) \cdot J\right)=J^{-1}\cdot {\bf I}^{-1}_\bfx(\psi) \cdot J^{-T}\cdot J^T \cdot {\bf II}_\bfx(\psi) \cdot J ,
\]
and Eq. \eqref{ch7-wein-gauss-mean2} is proved. Since from Eq. \eqref{ch7-wein-gauss-mean2} the matrices ${\bf W}_{\bfx\circ \psi}$, ${\bf W}_\bfx(\psi)$ are similar, they have the same determinant and trace, and Eq. \eqref{ch7-wein-gauss-mean3} follows. 
\end{proof}

Finally we reach the following \textcolor{black}{result}. 

\begin{proposition}\label{ch7-th-KH}
Let $\bfx(t,s)$ be a proper parametrization of a surface $S$, and let $f(\bfx)=A\bfx+\bfb$ be a Euclidean symmetry of $S$. Then there exists a birational planar mapping $\psi$ satisfying that $f\circ \bfx=\bfx\circ \psi$, and
\begin{equation}\label{ch7-eq-KH}
{\bf K}_{\bfx}={\bf K}_{\bfx}(\psi),\mbox{ }{\bf H}_\bfx=\det(A) {\bf H}_\bfx(\psi). 
\end{equation}
\end{proposition}

\begin{proof} From \textcolor{black}{the diagram in Eq. \eqref{ch7-eq:fundamentaldiagram}}, there exists a birational mapping $\psi$ such that $f\circ \bfx=\bfx\circ \psi$. By Corollary \ref{ch7-cor-A-transform}, ${\bf K}_{f\circ \bfx}={\bf K}_\bfx$, and by Corollary \ref{ch7-wein-gauss-mean3}, ${\bf K}_{\bfx\circ \psi}={\bf K}_\bfx(\psi)$. Since $f\circ \bfx=\bfx\circ \psi$, we get the result for ${\bf K}_\bfx$. For ${\bf H}_\bfx$ the idea is the same. 
\end{proof}

Let us see how to derive an algorithm from \textcolor{black}{Proposition} \ref{ch7-th-KH} in order to compute the symmetries of $S$. First, we will use $\widehat{\bf H}_\bfx={\bf H}^2_\bfx$, which is a rational function, instead of ${\bf H}_\bfx$. Writing
\[
{\bf K}_\bfx(t,s)={\bf K}_\bfx(u,v),\mbox{ }\widehat{\bf H}_\bfx(t,s)=\widehat{\bf H}_\bfx(u,v),
\]
we observe that the above equalities are satisfied, under the hypotheses of \textcolor{black}{Proposition} \ref{ch7-th-KH}, for $u=\psi_1(t,s)$, $v=\psi_2(t,s)$. Next ${\bf K}_\bfx$ and $\widehat{\bf H}_\bfx$ are rational functions, so after clearing denominators, the above equations lead to 
\begin{equation}\label{ch7-eq-xi12}
\xi_1(t,s,u,v)=0,\mbox{ }\xi_2(t,s,u,v)=0,
\end{equation}
where $\xi_1,\xi_2$ are two polynomials in the variables $t,s,u,v$, the first one coming from ${\bf K}_\bfx$, the second one from $\widehat{\bf H}_\bfx$. By well-known properties of resultants, if $t,s,u,v$ satisfy Eq. \eqref{ch7-eq-xi12} then $t,s,u$ satisfy 
\begin{equation}\label{ch7-eq-eta1}
\eta_1(t,s,u)=0,\mbox{ with }\eta_1:=\mbox{Res}_v(\xi_1,\xi_2),
\end{equation}
and $t,s,v$ satisfy 
\begin{equation}\label{ch7-eq-eta2}
\eta_2(t,s,v)=0,\mbox{ with }\eta_2:=\mbox{Res}_u(\xi_1,\xi_2).
\end{equation}
Now let 
\begin{equation}\label{ch7-psi12}
\psi_1(t,s)=\dfrac{\psi_{1,n}(t,s)}{\psi_{1,d}(t,s)},\mbox{ }\psi_2(t,s)=\dfrac{\psi_{2,n}(t,s)}{\psi_{2,d}(t,s)},
\end{equation}
where $\psi_{i,n},\psi_{i,d}$, for $i=1,2$, are polynomials, and $\gcd(\psi_{i,n},\psi_{i,d})=1$ for $i=1,2$. Then we have the following result. 

\begin{theorem}\label{ch7-themainKH}
Let $\bfx(t,s)$ be a proper parametrization of a surface $S$, and let $f(\bfx)=A\bfx+\bfb$ be a Euclidean symmetry of $S$. Then there exists a birational planar mapping $\psi=(\psi_1,\psi_2)$, with components as in Eq. \eqref{ch7-psi12}, such that $\psi_{1,d}(t,s) u-\psi_{1,n}(t,s)$ is a factor of $\eta_1(t,s,u)$, and $\psi_{2,d}(t,s) v-\psi_{2,n}(t,s)$ is a factor of $\eta_2(t,s,v)$. 
\end{theorem}

\begin{proof} From \textcolor{black}{the diagram in Eq. \eqref{ch7-eq:fundamentaldiagram}}, there exists a birational mapping $\psi$ such that $f\circ \bfx=\bfx\circ \psi$. Then for any $(t,s)$ we have that $(t,s,\psi_1(t,s),\psi_2(t,s))$ satisfies Eq. \eqref{ch7-eq-xi12}, and therefore $(t,s,\psi_1(t,s))$ fulfills Eq. \eqref{ch7-eq-eta1}, and $(t,s,\psi_2(t,s))$ fulfills Eq. \eqref{ch7-eq-eta2}. Hence, all the points of the surface (in the $(t,s,u)$ space) defined by $\psi_{1,d}(t,s) u-\psi_{1,n}(t,s)=0$, are also points of the surface $\eta_1(t,s,u)=0$. Furthermore, since  $\gcd(\psi_{1,n},\psi_{1,d})=1$ the polynomial $\psi_{1,d}(t,s) u-\psi_{1,n}(t,s)$ is irreducible, so by Study's Lemma (see Section 6.13 of \cite{Fischer}), $\psi_{1,d}(t,s) u-\psi_{1,n}(t,s)$ divides $\eta_1(t,s,u)$. For $\psi_{2,d}(t,s) v-\psi_{2,n}(t,s)$ we argue in the same way. 
\end{proof}

Thus, if, say, $\eta_1(t,s,u)$ is not identically zero, we can look for the factors of $\eta_1$ which are linear in $u$, which provides the functions $\psi_1(t,s)$, substitute $u=\psi_1(t,s)$ into $\xi_1(t,s,u,v)$, $\xi_2(t,s,u,v)$, and recover $v$ from the linear factor in $v$ of the gcd of the polynomials in $t,s,v$ resulting from that substitution. We might also work with $\eta_2(t,s,v)$, but we do not need  to compute both resultants. Notice that in this process it is not necessary to know a priori the form of $\psi(t,s)$. Furthermore, once $\psi$ is known, we can compute $f$ from the equality $f\circ \bfx=\bfx\circ \psi$. This leads to an algorithm, Algorithm \ref{ch7-alg-space}, for computing the symmetries of $S$. The algorithm works whenever not both resultants $\eta_1,\eta_2$ are identically zero; if $\eta_1,\eta_2$ are both identically zero then either $\xi_1,\xi_2$ share a factor, or some of them is identically zero. We can identify some cases where this can happen: 

\begin{itemize}
\item {\it Gauss curvature being constant:} It is well-known that the only surfaces with nonzero constant Gauss curvature ${\bf K}$ are isometric to the sphere, if ${\bf K}>0$, or to the \emph{pseudosphere}, if ${\bf K}<0$. However, the pseudosphere is a surface of revolution whose directrix is a trascendental curve, the \emph{tractix}, so the pseudosphere is not algebraic. This implies that the only irreducible algebraic surface with constant nonzero Gaussian curvature is the sphere. However, there are algebraic, and in fact rational, surfaces with vanishing Gauss curvature; these surfaces are called \emph{developable} surfaces, and will come up in the next section.  
\item {\it Mean curvature being constant:} These surfaces certainly exist, and among them we have the remarkable class of surfaces where ${\bf H}=0$, called \emph{minimal surfaces}. Symmetries of rational minimal surfaces are treated in \cite{AM22}.
\item {\it Certain functional relationships between the Gauss curvature and the mean curvature:} An example is the case of \emph{linear Weingarten surfaces} (see for instance \cite{PKPP21}), which are surfaces where $a{\bf H}+b{\bf K}=c$, for constants $a,b,c$; this family includes as particular cases surfaces with constant Gauss or mean curvature. 
\end{itemize}

In these cases the algorithm fails, and an alternative must be used; nevertheless, for developable surfaces we can use the ideas in the next section, and for minimal surfaces we can use \cite{AM22}.

\begin{algorithm}[t!]
\begin{algorithmic}[1]
\Require A proper parametrization $\bfx(t,s)$ of a surface $S\subset {\Bbb R}^3$.
\Ensure $\mbox{Sym}(S)$.
\State{compute ${\bf K}(t,s)$, $\widehat{\bf H}(t,s)$, and the polynomials $\xi_1,\xi_2$}
\State{compute the resultant $\eta_1(t,s,u)=\mbox{Res}_v(\xi_1,\xi_2)$}
\If{$\eta_1$ is not identically zero}
\State{compute the univariate factors in $u$ of $\eta_1(t,s,u)$, and the corresponding functions $u=\psi^{(j)}_1(t,s)$}
\For{each function $\psi^{(j)}_1(t,s)$}
\State{substitute $u=\psi^{(j)}(t,s)$ into $\xi_1,\xi_2$}
\State{compute the gcd of the polynomials, after the above substitution}
\State{compute the univariate factors in $v$ of the gcd}
\State{compute $v=\psi^{(j)}_1(t,s)$ by solving for $v$ in each univariate factor}
\EndFor
\Else
\State{compute the resultant $\eta_2(t,s,v)=\mbox{Res}_u(\xi_1,\xi_2)$}
\If{$\eta_2$ is not identically zero}
\State{proceed as before, with $\eta_2$}
\EndIf
\EndIf
\If{$\eta_1,\eta_2$ are both identically zero}
\State{{\bf return} {\tt method fails: zero resultants}}
\EndIf
\For{each couple $\psi^{(j)}(t,s)=(\psi^{(j)}_1(t,s),\psi^{(j)}_2(t,s))$ computed in the previous process}
\State{compute the symmetry $f$ from $f\circ \bfx=\bfx\circ \psi^{(j)}$ and {\bf return} $f$}
\EndFor
\end{algorithmic}
\caption{Rational Surface Symmetries}\label{ch7-alg-space}
\end{algorithm}

\begin{example}
Consider the toy example of the ellipsoid parametrized by 
\[
\bfx(t,s)=\left(\dfrac{2(-s^2 - t^2 + 1)}{s^2 + t^2 + 1}, \dfrac{-2t}{s^2 + t^2 + 1}, \dfrac{8s}{s^2 + t^2 + 1}\right).
\]
In this case, Algorithm 4 provides 8 symmetries, including the trivial symmetry (the identity)\textcolor{black}{,} corresponding to the Cremona transformations:
\[
\left(\dfrac{t}{t^2+s^2},\frac{-s}{t^2+s^2}\right),\mbox{ } \left(\dfrac{t}{t^2+s^2},\frac{s}{t^2+s^2}\right), \mbox{ }
\left(\dfrac{-t}{t^2+s^2},\frac{-s}{t^2+s^2}\right),\mbox{ } \left(\dfrac{-t}{t^2+s^2},\frac{s}{t^2+s^2}\right),
\]
and 
\[
(-t,-s),\mbox{ }(-t,s),\mbox{ }(t,-s),\mbox{ }(t,s).
\]
For instance, the symmetry corresponding to the first Cremona transformation, $\left(\dfrac{t}{t^2+s^2},\dfrac{-s}{t^2+s^2}\right)$, is $f(\bfx)=A\bfx$ with 
\[
A=\begin{bmatrix}
-1 & 0 & 0 \\
0 & 1 & 0\\
0 & 0 & -1
\end{bmatrix}.
\]
Notice in particular that the above Cremona transformation is not linear projective, i.e. its components are not quotients of linear polynomials in $t,s$. 
\end{example}

\section{Symmetries of ruled surfaces}\label{sec-specruled}

\subsection{Background on ruled surfaces}

We say that $S\subset {\Bbb R}^3$ is \emph{ruled} if $S$ is covered by lines, called the \emph{rulings} of the surface. These surfaces are well-known in Differential Geometry; one can check classical texts like \cite{DoCarmo, Gray99, Struik} for further reading on the topic, and also for several facts that we will be using in this subsection.

\par
If $S$ is rational, then  we can always find \cite{PDS14} a rational parametrization of the form 
\begin{equation}\label{ch7-standard}
\bfx(t,s)=\bfu(t)+s\bfv(t),
\end{equation}
where $\bfu(t),\bfv(t)$ parametrize space rational curves. The curve defined by $\bfu(t)$ is called the \emph{directrix}. We refer to Eq. \eqref{ch7-standard} as a \emph{standard parametrization} of $S$, and we will assume the ruled surface $S$ we work with to be parametrized in this way; notice that, regardless of potential reparametrizations of $\bfu(t),\bfv(t)$, a standard parametrization of $S$ is not necessarily unique, since different curves may serve as directrices of the surface. 

\par
A special type of ruled surfaces, called \emph{developable surfaces}, is classical in Differential Geometry, and widely used in applications. Intuitively speaking, developable surfaces are the surfaces that can be ``unfolded" onto the plane, so that we can find a dipheomorphic mapping $\phi:S\to {\Bbb R}^2$ preserving distances.  

\par
Developable surfaces can be characterized as the ruled surfaces with vanishing Gauss curvature. Also, if $S$ is parametrized by Eq. \eqref{ch7-standard} then $S$ is developable if and only if the mixed product $[\bfu'(t),\bfv(t),\bfv'(t)]$ is identically zero. This allows us to classify developable surfaces in three types: 
\begin{itemize}
\item [(i)] \emph{Cylindrical surfaces}: Ruled surfaces where all the rulings are parallel to a same vector. Thus, they can be parametrized as 
\begin{equation}\label{ch7-cyl-param}
\bfx(t,s)=\bfu(t)+s\bfv_0,
\end{equation}
where $\bfv_0$ is a constant vector. 
\item [(ii)] \emph{Conical surfaces}: These are ruled surfaces where all the rulings intersect at a same point, called the \emph{vertex} of the surface. By applying if necessary a translation so that the vertex coincides with the origin, these surfaces can be parametrized as 
\begin{equation} \label{ch7-con-param}
\bfx(t,s)=s\bfv(t).
\end{equation}
\item [(iii)] \emph{Tangential surfaces}: Ruled surfaces consisting of the union of tangent lines to the directrix, which can be parametrized as 
\begin{equation}\label{ch7-tang}
\bfx(t,s)=\bfu(t)+s\bfu'(t).
\end{equation}
\end{itemize}

\par 
Cylindrical surfaces and conical surfaces can be recognized, for example, using the algorithm in \cite{AG17}; furthermore, in those cases we can find the constant direction of the rulings, for cylindrical surfaces, and the vertex, for conical surfaces. 

\subsection{Computation of symmetries of ruled surfaces}

In order to compute the symmetries of a ruled surface, we also need to recall the notion of a \emph{line of striction} (see for instance Section 3-5 of \cite{DoCarmo} for further information on this notion, and its properties). Given a ruled surface $S$, parametrized as in Eq. \eqref{ch7-standard}, the line of striction ${\mathcal E}$ is a parametric space curve $\bfc(t)$, contained in $S$, such that $\bfc'(t)$ is orthogonal to $\bfv'(t)$. The expression for the line of striction is
\begin{equation}\label{strictionline}
\bfc(t)=\bfu(t)-\dfrac{(\bfv(t)\times \bfv'(t))\cdot (\bfv(t)\times \bfu'(t))}{\Vert \bfv(t)\times \bfv'(t)\Vert^2}\bfv(t),
\end{equation}
and can be proven to be independent of the directrix considered in Eq. \eqref{ch7-standard}. In particular, if Eq. \eqref{ch7-standard} is rational then $\bfc(t)$ is also rational. This expression does not make sense for cylindrical surfaces; also, for conical surfaces the line of striction degenerates into a point. Excluding cylindrical surfaces, we can always assume that $S$ is parametrized as in Eq. \eqref{ch7-standard}, where the directrix is the line of striction. One can check that for tangential developable surfaces the directrix appearing in Eq. \eqref{ch7-tang} is already the line of striction. 

\par
The reason why the line of striction is useful in our context is that any symmetry of the surface is also a symmetry of the line of striction. In order to see this, let the directrix $\bfu(t)$ be the line of striction, and let $f(\bfx)=A\bfx+\bfb$ be a symmetry of $S$. Then 
\begin{equation}\label{ch7-sym-striction}
\tilde{\bfx}(t,s)=(f\circ \bfx)(t,s)=A\bfx(t,s)+\bfb=\underbrace{(A\bfu(t)+\bfb)}_{\tilde{\bfu}(t)}+s\underbrace{A\bfv(t)}_{\tilde{\bfv}(t)}
\end{equation}
also parametrizes $S$, and we can derive the following theorem. 

\begin{theorem}\label{ch7-lem-striction}
Let $f(\bfx)=A\bfx+\bfb$ be a symmetry of a ruled surface $S$, and let ${\mathcal E}$ be the line of striction of $S$. Then $f$ is also a symmetry of ${\mathcal E}$. 
\end{theorem}

\begin{proof}
Assume that the directrix $\bfu(t)$ is the line of striction ${\mathcal E}$. Let us see that the parametrization $\tilde{\bfu}(t)$ in Eq. \eqref{ch7-sym-striction}, which parametrizes the image of ${\mathcal E}$ under $f$, is also a parametrization of ${\mathcal E}$. Indeed, 
\[
\tilde{\bfu}'(t)\cdot \tilde{\bfv}'(t)=(A\bfu'(t))\cdot (A\bfv'(t)).
\]
Since $A$ is an orthogonal matrix, $(A\bfu'(t))\cdot (A\bfv'(t))=\bfu'(t)\cdot \bfv'(t)=0$. Thus, $\tilde{\bfu}(t)$ satisfies the condition defining the line of striction, parametrized by $\bfu(t)$, so $\tilde{\bfu}(t)$ also parametrizes ${\mathcal E}$. Thus $f({\mathcal E})={\mathcal E}$, and the result follows. 
\end{proof}

The symmetries of the line of striction can be computed using known algorithms for space rational curves. Thus, whenever the line of striction has finitely many symmetries, we just need to test which ones are also symmetries of $S$. If the line of striction has infinitely many symmetries, i.e. if the line of striction is either a line or a circle, then we can use the algorithm in the previous section, or other algorithms known for ruled surfaces. Notice that cylindrical and conical surfaces are excluded in this approach.

\subsection{Developable surfaces}

Since developable surfaces have zero Gauss curvature, the general algorithm presented in Section \ref{gen-alg} is not applicable to them. Also, for two of the subfamilies of developable surfaces, cylindrical and conical surfaces, using the line of striction is not an option, either. 

\par
For cylindrical surfaces, the symmetries of the surface follow from the symmetries of a normal section of the surface, i.e. the intersection of the surface with a plane $\Pi$ normal to the direction ${\bf v}$ of the rulings. Because Eq. \eqref{ch7-standard} is linear in $s$, the intersection of $S$ with a plane $\Pi$ normal to the rulings is a rational planar curve, so its symmetries can be computed using known algorithms for rational curves. Notice also that cylindrical surfaces are invariant under translations by any vector parallel to ${\bf v}$. 

\par
For conical surfaces, parametrized as in Eq. \eqref{ch7-con-param}, we just need to consider symmetries fixing the origin, i.e. symmetries $f(\bfx)=A\bfx$, since in Eq. \eqref{ch7-con-param} the origin is the vertex of the surface, which must be kept invariant under any symmetry. Then we can just apply the method in Section 4.1 of \cite{AQ}. Notice that any symmetry of the curve defined by $\bfv(t)$ (see Eq. \eqref{ch7-con-param}) will also be a symmetry of the surface, but the converse is false (e.g. a cone of revolution, with $\bfv(t)$ corresponding to an ellipse contained in the cone). 

\par
Finally, in the case of tangential developable surfaces Theorem \ref{ch7-lem-striction} works without any problem. However, we can go a bit further. 

\begin{proposition}\label{ch7-symmofconic2}
Let $S$ be a rational tangential developable surface. Then $f:{\Bbb R}^3 \to {\Bbb R}^3$ is a symmetry of $S$ if and only if $f$ is a symmetry of its line of striction. 
\end{proposition}

\begin{proof}
The implication $(\Rightarrow)$ is Theorem \ref{ch7-lem-striction}, so let us address $(\Leftarrow)$. Thus, let $f:{\Bbb R}^3 \to {\Bbb R}^3$, $f(\bfx)=A\bfx+\bfb$ be a symmetry of the line of striction of $S$; let $\bfu(t)$ be a rational parametrization of the line of striction, and let us assume without loss of generality that $\bfu(t)$ is proper. By \textcolor{black}{Proposition} \ref{ch7-fund-th-curves}, $A\bfu(t)+\bfb=(\bfu\circ \varphi)(t)$, with $\varphi(t)$ a M\"obius transformation. Considering the parametrization of $S$ provided by Eq. \eqref{ch7-tang}, 
\begin{equation}\label{ch7-dev-first}
(f\circ \bfx)(t,s)=(A\bfu(t)+\bfb)+sA\bfu'(t)=(\bfu\circ \varphi)(t)+sA\bfu'(t).
\end{equation}
Differentiating $A\bfu(t)+\bfb=(\bfu\circ \varphi)(t)$ and using the Chain Rule, we deduce that $A\bfu'(t)=\bfu'(\varphi(t))\varphi'(t)$. Substituting this into Eq. \eqref{ch7-dev-first}, we obtain
\begin{equation}\label{ch7-dev-second}
(f\circ \bfx)(t,s)=\bfu(\varphi(t))+s\bfu'(\varphi(t))\varphi'(t).
\end{equation}

The above equation can be written as
\begin{equation}\label{ch7-dev-second*}
(f\circ \bfx)(t,s)=\bfu(\varphi(t))+\psi_2(t,s)\bfu'(\varphi(t))=\bfx(\varphi(t),\psi_2(t,s)),
\end{equation}
with $\psi_2(t,s)=s\varphi'(t)$. Eq. \eqref{ch7-dev-second*} means, by \textcolor{black}{the diagram in Eq. \eqref{ch7-eq:fundamentaldiagram}}, that $f$ is a symmetry of $S$.
\end{proof}

Observe that for tangential surfaces the line of striction $\bfu(t)$ cannot be planar, since in that case the surface $S$ is a plane, which is a trivial case. Thus, the number of symmetries of $\bfu(t)$ must be finite. 

\section{Experimentation}\label{sec-experimentation}

In this section we provide some experimentation on the \textcolor{black}{algorithms given in the previous sections}. We have implemented \textcolor{black}{both algorithms} in the CAS Maple \cite{maple}, and tried several examples, in an Intel(R) Core(TM) $i7$ with 3.6 GHz processor and 32 Gb RAM. All the examples and implementation can be found in the last author's personal website \cite{website}\textcolor{black}{, and also in \cite{zenodo}. We start with the general algorithm described in Section \ref{gen-alg}; then we move to ruled surfaces and the algorithm given in Section \ref{sec-specruled}.}

\subsection{General algorithm}

\textcolor{black}{In order to test the algorithm for the general case we consider several surfaces of three different types: \emph{toric} surfaces, \emph{PN-surfaces}, and \emph{conoids}; in all these cases} the surfaces have projective base points, so the algorithm in \cite{HJS19} is not applicable. \textcolor{black}{Additionally, we have also considered some examples from the papers \cite{HJS19,JLS22}.}

\begin{itemize}
\item[] {\bf Torics.} A toric \textcolor{black}{surface} is a surface parametrized by $\bfx(t,s)=(t^{m_1}s^{n_1},t^{m_2}s^{n_2},t^{m_3}s^{n_3})$, where $m_i,n_i\in\mathbb{Z}$. In Table \ref{T1}, we provide the computation time ($t$, in seconds) for detecting symmetries of toric surfaces with various degrees.

\begin{table}[H]
\centering
{\renewcommand{\arraystretch}{1.7}
\begin{tabular}{c l c r}
\hline
Degree & Parametrization & Symmetries  & $t$\\
\hline
$2$ & $\left( t^2,\frac{t}{s},s\right)$ & $8$ & $0.125$  \\
$3$ & $\left( \frac{1}{s^2t},\frac{s}{t^2},\frac{1}{s}\right)$ & $4$ & $3.391$  \\
$4$ & $\left( t^3s,\frac{t^3}{s},s\right)$ & $12$ & $0.156$  \\
$5$ & $\left( t^4s,\frac{t^5}{s},s\right)$ & $4$ & $0.188$  \\
$6$ & $\left( t^5s,\frac{t^2}{s},s\right)$ & $4$ & $0.094$  \\
$7$ & $\left( t^5s^2,t^3,s^3\right)$ & $12$ & $1.594$  \\
$8$ & $\left( t^5s^3,t^8,s^3\right)$ & $12$ & $0.094$  \\
$9$ & $\left( s,\frac{t^2}{s},t^9\right)$ & $8$ & $0.734$\\ 
\hline
\end{tabular}
}
\caption{CPU times $t$ (seconds) for toric parametrizations of various degrees.}\label{T1}
\end{table}

\item[] {\bf PN-surfaces.} A PN-surface is a rational surface with a rational normal vector field. In the case of PN-surfaces, we replace in our algorithm the equation $\mbox{ }\widehat{\bf H}_\bfx(t,s)=\widehat{\bf H}_\bfx(u,v)$ by $\mbox{ }\bf H_\bfx(t,s)=\pm\bf H_\bfx(u,v)$. This reduces the degree of the equations in the resultants. \textcolor{black}{In order to} generate non-trivial PN-surfaces we used the \textcolor{black}{strategy} given in \cite{KMV}, where the authors provide a method using quaternions. First we consider the cubic PN-surface (See Fig \ref{fig:second}) given by the following rational parametrization.

\begin{equation*}
\bfx(t,s)=\begin{pmatrix}
-\frac{1}{3} s^3 -\frac{1}{3} s^2t +\frac{7}{3} st^2 + t^3 -\frac{5}{3} s^2 +\frac{26}{3} st +\frac{25}{3} t^2 \\[3pt]
-\frac{1}{3} s^3 -\frac{5}{3} s^2t -\frac{5}{3} st^2 + t^3 -\frac{13}{3} s^2 -10 st +\frac{13}{3} t^2 -\frac{46}{3} s +\frac{2}{3} t\\[3pt]
-\frac{4}{3} s^3 - 4 s^2t -\frac{4}{3} st^2 - 4t^3 - 9s^2 +\frac{14}{3} st +\frac{5}{3} t^2 -\frac{10}{3} s +\frac{26}{3} t
\end{pmatrix}
\end{equation*}

This PN-surface has \textcolor{black}{one} projective base point, $[-3:1:0]$, and admits two symmetries corresponding to the Cremona transformations
\begin{equation*}
\left(\dfrac{3}{5}t-\dfrac{4}{5}s+\dfrac{1}{5},\dfrac{3}{5}t-\dfrac{4}{5}s-\dfrac{3}{5}\right), (t,s),
\end{equation*}
with symmetries
\begin{align*}
f_1(\bfx)&=\begin{pmatrix}
\frac{24}{25} & \frac{7}{25} & 0\\[4pt]
\frac{7}{25} & -\frac{24}{25} & 0\\[4pt]
0 & 0 & 1
\end{pmatrix}\bfx+\begin{pmatrix}
-\frac{98}{75}\\[4pt]
\frac{686}{75}\\[4pt]
0
\end{pmatrix}\\
f_2(\bfx)&=\begin{pmatrix}
1 & 0 & 0\\
0 & 1 & 0\\
0 & 0 & 1
\end{pmatrix}\bfx+\begin{pmatrix}
0\\
0\\
0
\end{pmatrix}
\end{align*}

The whole computation of this example took $0.203$ seconds. Now, \textcolor{black}{consider the} cubic PN-surface (See Fig \ref{fig:first}) given by the following rational parametrization.
\begin{equation*}
\bfy(t,s)=\begin{pmatrix}
-\frac{4}{13} s +\frac{1}{13} s^2 -\frac{11}{156} s^3 -\frac{3}{32} s^4 - \frac{4}{13} t -\frac{6}{13} ts -\frac{9}{52} ts^2 -\frac{7}{13} t^2 -\frac{9}{52} t^2s -\frac{3}{16} t^2s^2 -\frac{59}{156} t^3 -\frac{3}{32} t^4\\[3pt]
\frac{4}{13} s^2 -\frac{11}{52} s^3 +\frac{1}{8} s^4 +\frac{25}{52} ts^2 +\frac{4}{13} t^2 -\frac{1}{52} t^2s +\frac{1}{4} t^2s^2 +\frac{19}{52} t^3 +\frac{1}{8} t^4\\[3pt]
-\frac{4}{13}s + s^2 -\frac{17}{78} s^3 +\frac{3}{8} s^4 +\frac{28}{13} t -\frac{8}{13} ts +\frac{22}{13} ts^2 +\frac{37}{13} t^2 -\frac{4}{13} t^2s +\frac{3}{4} t^2s^2 +\frac{131}{78} t^3 +\frac{3}{8} t^4
\end{pmatrix}
\end{equation*}

This PN-surface has two projective base points $[\mathbf{i}:1:0]$ and $[-\mathbf{i}:1:0]$. The surface admits two symmetries corresponding to the Cremona transformations
\begin{equation*}
\left(-s-\dfrac{12}{13},-t-\dfrac{12}{13}\right), (t,s),
\end{equation*}
with symmetries
\begin{align*}
f_1(\bfx)&=\dfrac{1}{169}\begin{pmatrix}
137 & -96 & 24\\
-96 & -119 & 72\\
24 & 72 & 151
\end{pmatrix}\bfx+\dfrac{1}{28561}\begin{pmatrix}
3840\\
11520\\
-2880
\end{pmatrix}\\
f_2(\bfx)&=\begin{pmatrix}
1 & 0 & 0\\
0 & 1 & 0\\
0 & 0 & 1
\end{pmatrix}\bfx+\begin{pmatrix}
0\\
0\\
0
\end{pmatrix}
\end{align*}

The whole computation of this example took $0.125$ seconds.

\begin{figure}
\centering
\begin{subfigure}{0.4\textwidth}
    \includegraphics[width=\textwidth]{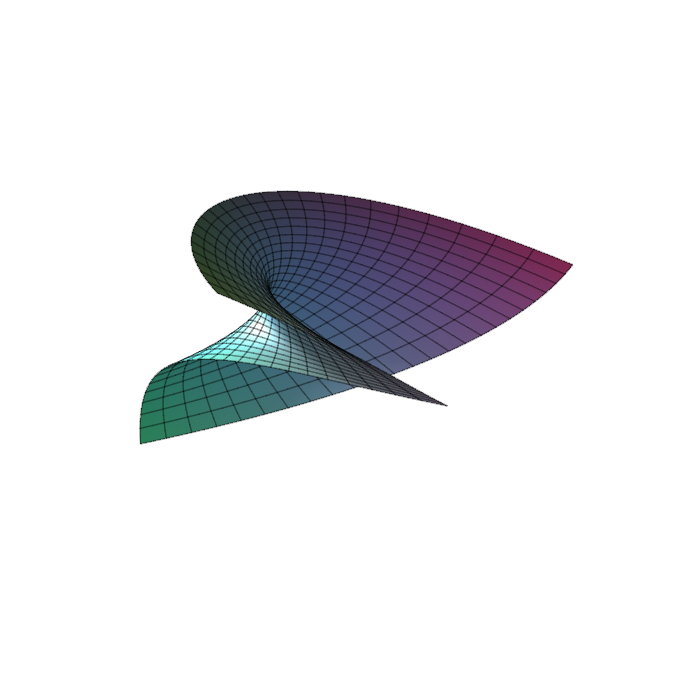}
    \caption{Cubic PN-surfece}
    \label{fig:first}
\end{subfigure}
\begin{subfigure}{0.4\textwidth}
    \includegraphics[width=\textwidth]{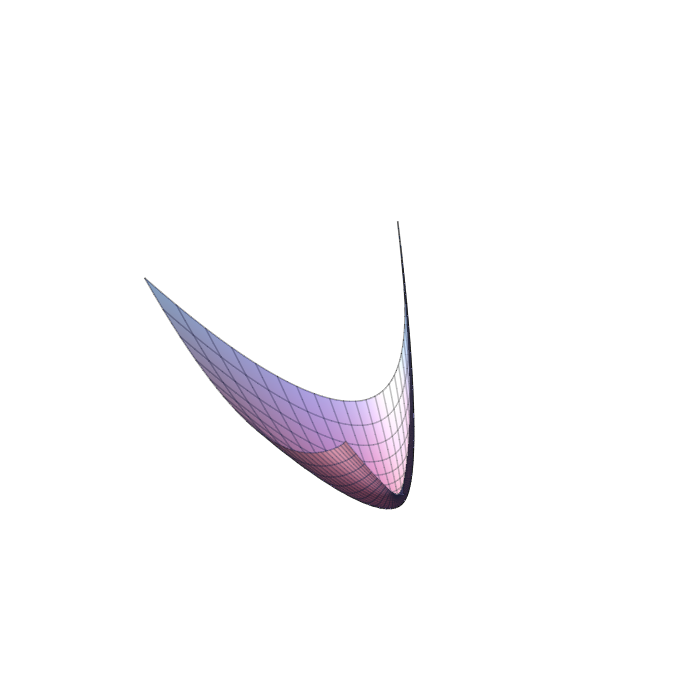}
    \caption{Quartic PN-surface}
    \label{fig:second}
\end{subfigure}
\caption{Plots of the PN-Surfaces}
\label{fig:figures}
\end{figure}
\item[] {\bf Plücker's conoids.} The Plücker's conoid is a ruled surface defined by the following parametrization
\begin{equation*}
\bfx(r,\theta)=(r\cos\theta,rsin\theta,\sin 2\theta).
\end{equation*} 

In our tests, we used a generalization of the Plücker's conoid, which is again a ruled surface, defined by
\begin{equation*}
\bfx(r,\theta)=(r\cos\theta,rsin\theta,\sin 2n\theta),
\end{equation*}
where $n>0$. Using the fact that
\begin{equation*}
(\cos\theta,sin\theta)\longleftrightarrow\left(\dfrac{1-t^2}{1+t^2},\dfrac{2t}{1+t^2}\right)
\end{equation*}
and exploiting the Chebychev polynomials of the second kind $U_k$, we get a rational parametrization of the generalized Plücker's conoid as
\begin{equation*}
\bfx(t,s)=\left(\dfrac{(1-t^2)s}{1+t^2},\dfrac{2ts}{1+t^2},\dfrac{2t}{1+t^2}U_k\left(\dfrac{1-t^2}{1+t^2}\right) \right),
\end{equation*}
where $k>0$ (see Fig \ref{fig:P} for the plots of the surfaces for $k\in\{1,2,3,4\}$). In Table \ref{T2}, we provide the computation times for detecting symmetries of the generalized Plücker's conoids with various degrees. Notice that any generalized Plücker's conoid has two projective base points, $[1:0:0]$ and $[0:1:0]$.

\begin{table}[H]
\centering
{\renewcommand{\arraystretch}{1.7}
\begin{tabular}{c c r}
\hline
Degree & Symmetries  & $t$\\
\hline
$4$ & $16$ & $0.360$  \\
$6$ & $8$ & $1.203$  \\
$8$ & $16$ & $6.407$  \\
$10$ &  $8$ & $16.812$  \\
$12$ &  $16$ & $41.500$  \\
$14$ &  $8$ & $107.391$  \\
$16$ &  $16$ & $162.469$ \\
$18$  & $8$ & $195.578$\\
\hline
\end{tabular}
}
\caption{CPU times $t$ (seconds) for the generalized Plücker's conoids of various degrees.}\label{T2}
\end{table} 

\begin{figure}
\centering
\begin{subfigure}{0.2\textwidth}
    \includegraphics[width=\textwidth]{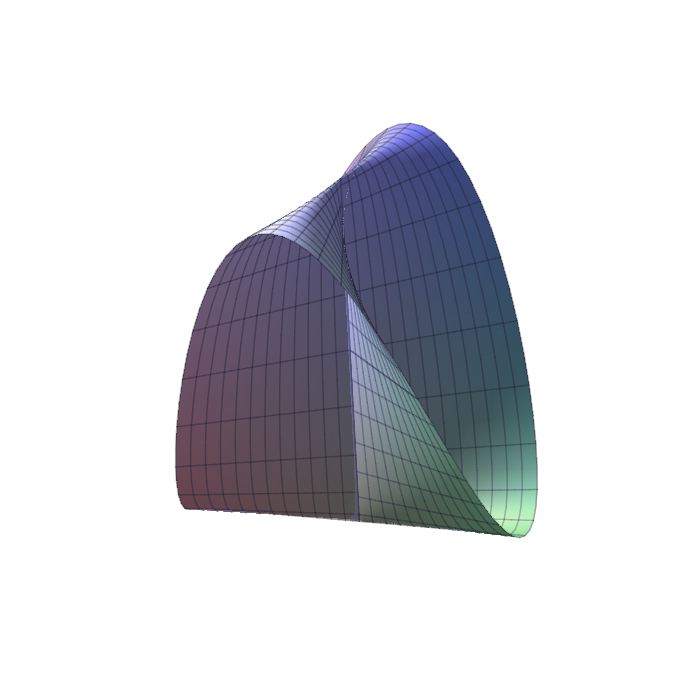}
    \caption{$k=1$}
    \label{fig:1}
\end{subfigure}
\begin{subfigure}{0.2\textwidth}
    \includegraphics[width=\textwidth]{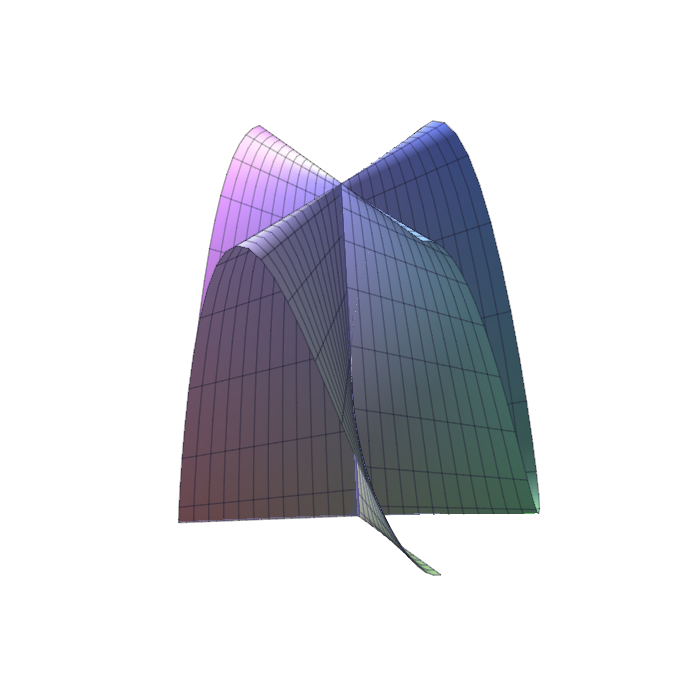}
    \caption{$k=2$}
    \label{fig:2}
\end{subfigure}
\begin{subfigure}{0.2\textwidth}
    \includegraphics[width=\textwidth]{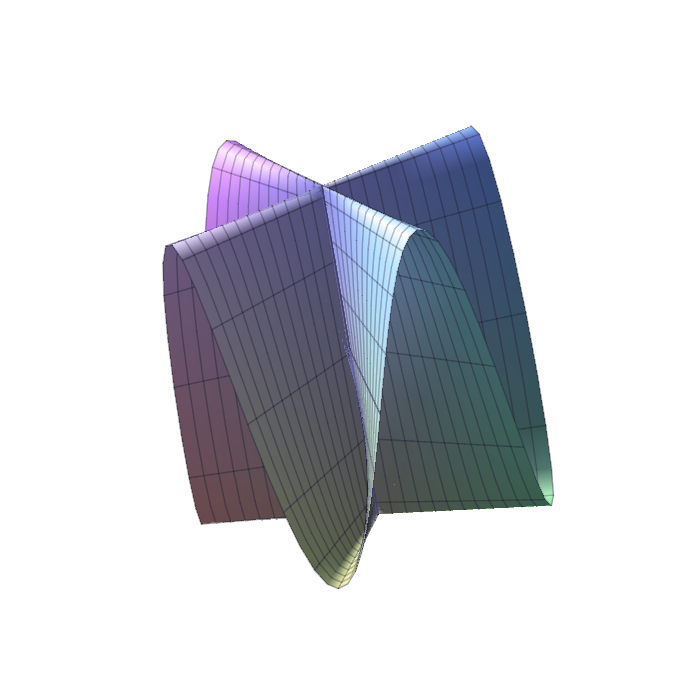}
    \caption{$k=3$}
    \label{fig:3}
\end{subfigure}
\begin{subfigure}{0.2\textwidth}
    \includegraphics[width=\textwidth]{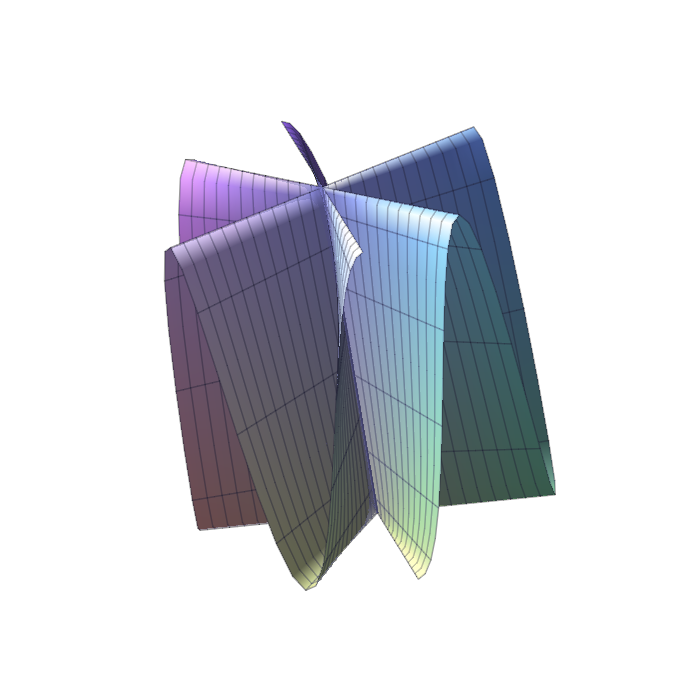}
    \caption{$k=4$}
    \label{fig:4}
\end{subfigure}
\caption{Generalized Plücker's conoids for different values of $k$}
\label{fig:P}
\end{figure}

\begin{color}{black}
\item [] {\bf Some examples from \cite{HJS19,JLS22}}

We have tried examples of \emph{quadratically parametrizable surfaces}, considered in Section 5.1 of \cite{HJS19}; following the notation in \cite{HJS19}, they are denoted as $\Sigma_1,\ldots,\Sigma_9$, where $\Sigma_1,\ldots,\Sigma_6$ are base-point free (so the algorithm in \cite{HJS19} can be applied), and $\Sigma_7,\Sigma_8,\Sigma_9$ are not (and the algorithm in \cite{HJS19} cannot be applied). We list the timings in Table \ref{quada}; we can also find timings in \cite{HJS19}, although in the case of \cite{HJS19} the algorithm finds projective mappings leaving the surface invariant, and not only isometries. For $\Sigma_1,\Sigma_2$, the timings in \cite{HJS19} are much better; for $\Sigma_3$ to $\Sigma_6$, the timings for our algorithm and the one in \cite{HJS19} are comparable; $\Sigma_7,\Sigma_8,\Sigma_9$, however, cannot be solved by \cite{HJS19}. In all the cases, the parametrizations have small coefficients. 

\begin{table}[H]
\centering
\begin{tabular}{l r}
The curve & Timing in seconds\\
\hline
$\Sigma_1$ &  $50.016$   \\
$\Sigma_2$ &  $70.344$   \\
$\Sigma_3$ &  $0.140$   \\
$\Sigma_4$ &  $0.235$   \\
$\Sigma_5$ &  $0.688$   \\
$\Sigma_6$ &  $0.001$   \\
$\Sigma_7$ &  $1.719$   \\
$\Sigma_8$ & $0.053$ \\
$\Sigma_9$ & $0.023$ \\
\hline
\end{tabular}
\caption{Quadratically parametrizable surfaces (see \cite{HJS19})}\label{quada}
\end{table}

We have also tried some concrete examples of quadratically parametrized surfaces appearing in \cite{HJS19} (see Table 6 in \cite{HJS19}). We provide the timings in Table \ref{quada2}; the first column reproduces the classification as spelt in \cite{HJS19}.

\begin{table}[H]
\centering
\begin{tabular}{l l r}
Class of the curve & The curve & Timing\\
\hline
$3-2-1a$ & $(t^2,s^2,t+s)$ & $1.243$ \\
$3-2-1c$ & $(t^2,s^2+t,s)$ & $0.078$ \\
$3-2-3$  & $(t^2-s^2,ts,t)$ & $0.363$ \\
$3-3-1b$ & $(t^2,s^2,ts+t)$ & $5.327$ \\
$3-3-1c$ & $(t^2,s^2,t+s+ts)$ & $9.243$ \\
$3-3-2a$ & $(t^2,s^2+t,ts)$ & $0.0781$ \\
$3-3-2b$ & $(t^2,s^2,ts-s)$ & $0.095$ \\
\end{tabular}
\caption{Some concrete $\Sigma_i$s} \label{quada2}
\end{table}

Additionally, we also computed the symmetries of the \emph{Roman Surface}, 
\begin{equation*}
	\bfx (t,s)=\begin{pmatrix}
		\dfrac{t}{1+t^2+s^2},
		\dfrac{s}{1+t^2+s^2},
		\dfrac{ts}{1+t^2+s^2}
	\end{pmatrix}
\end{equation*}
which appears as Example 5 of \cite{JLS22}. The computation of its 8 symmetries took $50.016$ seconds with our algorithm; there are no timings in \cite{JLS22}, so here we cannot compare. For other examples in \cite{JLS22}, the computations with our algorithm were too costly.
\end{color}

\end{itemize}

\subsection{Ruled surfaces}

Here we consider the method in Section \ref{sec-specruled}, which requires to determine first the line of striction of the surface. Then we apply the method in \cite{AHM15} to find the symmetries of the rational curve defined by the line of striction. In Table \ref{tab:cmp}, we compare the timing ($t_r$) of the algorithm for ruled surfaces in Section \ref{sec-specruled}, the timing ($t_e$) of the algorithm given in \cite{AQ}, and the the timing ($t_g$) for the general algorithm given in Section \ref{gen-alg}. One can see that the method in Section \ref{sec-specruled} beats the other methods in all the cases. 

\begin{table}[H]
\centering
  \begin{tabular}{cccccc} 
  \hline Param. & Deg. & Symm. & $t_r$ & $t_e$ & $t_g$ \\
      \hline 
      $\bfx_1$ & $9$ & $8$ & \textcolor{black}{$0.156$} & $9.640$ & $>10^3$  \\
      $\bfx_2$ & $7$ & $1$ & \textcolor{black}{$0.297$} & $1.981$ & $>10^3$   \\
     $\bfx_3$ & $7$ & $2$ & \textcolor{black}{$0.563$} & $1.888$ & $>10^3$   \\
      $\bfx_4$ & $5$ & $2$ & \textcolor{black}{$0.016$} & $1.684$ & $902.213$  \\
      $\bfx_5$ & $2$ & $2$ & \textcolor{black}{$0.047$} & $3.448$ & $54.141$   \\
      $\bfx_6$ & $7$ & $2$ & \textcolor{black}{$0.828$} & $1.935$ & $986.213$  \\
      $\bfx_7$ & $6$ & $2$ & \textcolor{black}{$0.063$} & $1.716$ & $626.012$   \\
      $\bfx_8$ & $17$ & $8$ & \textcolor{black}{$0.172$} & $9.828$ & $>10^3$  \\
      \hline
  \end{tabular}
  \caption{CPU times (seconds) for symmetries of ruled surfaces given in Table \ref{tab:surf}} \label{tab:cmp}
\end{table}

\begin{remark}
The Pl\"ucker conoids analyzed in the previous subsection are also ruled surfaces, but their striction lines are straight lines, and therefore the ideas of Section \ref{sec-specruled} are not useful for that case. 
\end{remark}

\begin{remark}
In Table \ref{tab:cmp} we have not included the comparison with \cite{BLV}, where no timings are provided. But in any case, as we mentioned in the Introduction to this paper, the algorithm in \cite{BLV} requires using rational curves in higher dimension (${\Bbb P}^5$) plus a further step involving the solution of a polynomial quadratic system in a high number of variables, namely 7. Thus, one can expect that the timings of \cite{BLV} are, in general, much higher than ours. Recall, though, that the algorithm in \cite{BLV} is aimed to the more general problem of finding projective equivalences between ruled surfaces. 
\end{remark}

\begin{table}[H]
  \begin{tabular}{lll} 
  \hline Parametrization & Line of striction & Plots of Surf. \\
      \hline 
      $\scalemath{0.7}{\bfx_1(t,s)=\begin{pmatrix}
      \dfrac{2t^8 - 10t^6 - 10t^4 + 5t^2 + 1}{t^2+1}+s(2t^5-12t^3+2t)\\[6pt]
      -\dfrac{t^9 - 6t^7 +6t^3 + t^2 -3t+1}{t^2+1}+s(-t^6+7t^4-7t^2+1)\\[6pt]
      t^7+3t^5+3t^3+t+5+s(t^2+1)^3
      \end{pmatrix}}$ & \parbox[c]{1em}{\includegraphics[width=1in]{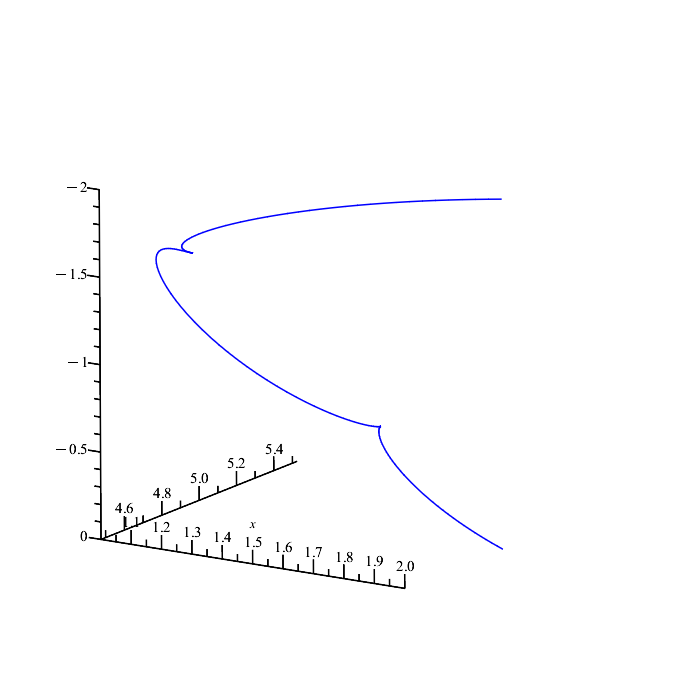}} & \parbox[c]{1em}{\includegraphics[width=1in]{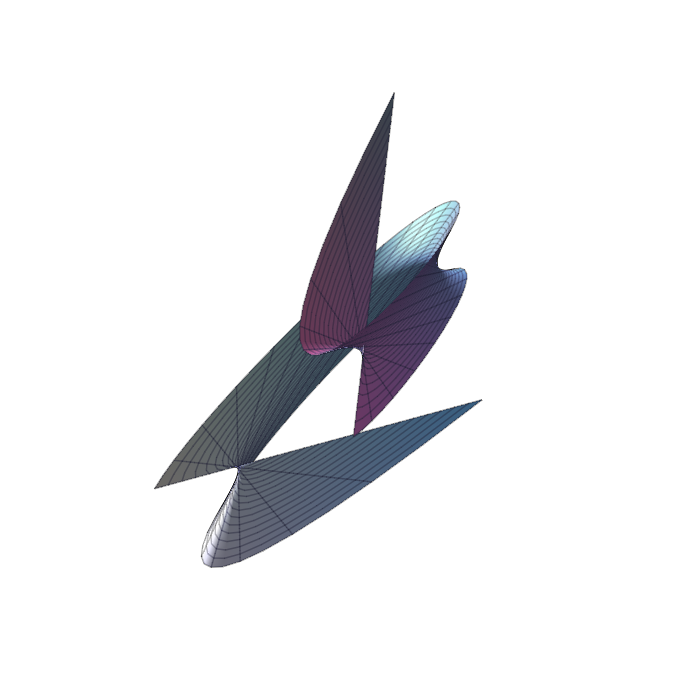}} \\
       $\scalemath{0.7}{\bfx_2(t,s)=\begin{pmatrix}
      \dfrac{t^7 + 7t^5 + 3t^3 - t^2 - 3t + 1}{t^2+1}+s(-t^4 - 6t^2 + 3)\\[6pt]
      \dfrac{8t^6 +8t^4 +2t}{t^2+1}+8st^3\\[6pt]
      t(t^2+1)+2+s(t^2+1)^2
      \end{pmatrix}}$ & \parbox[c]{1em}{\includegraphics[width=1in]{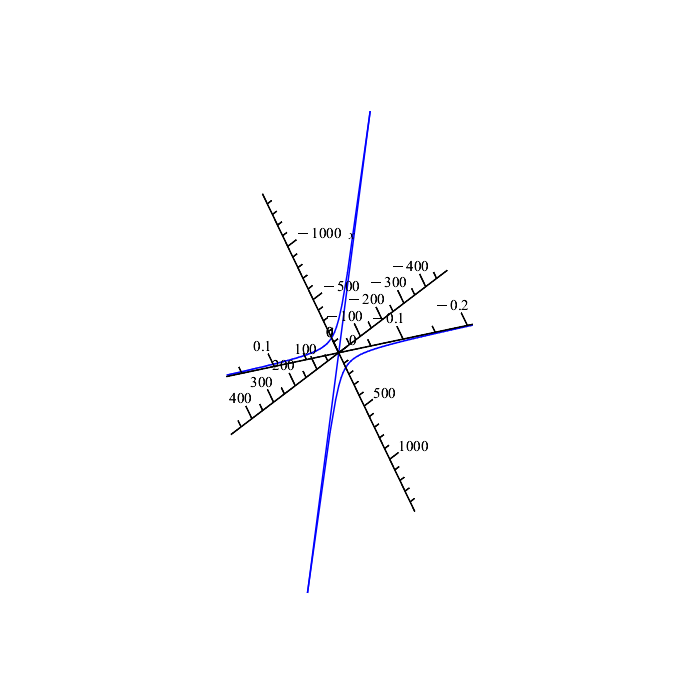}}&\parbox[c]{1em}{
      \includegraphics[width=1in]{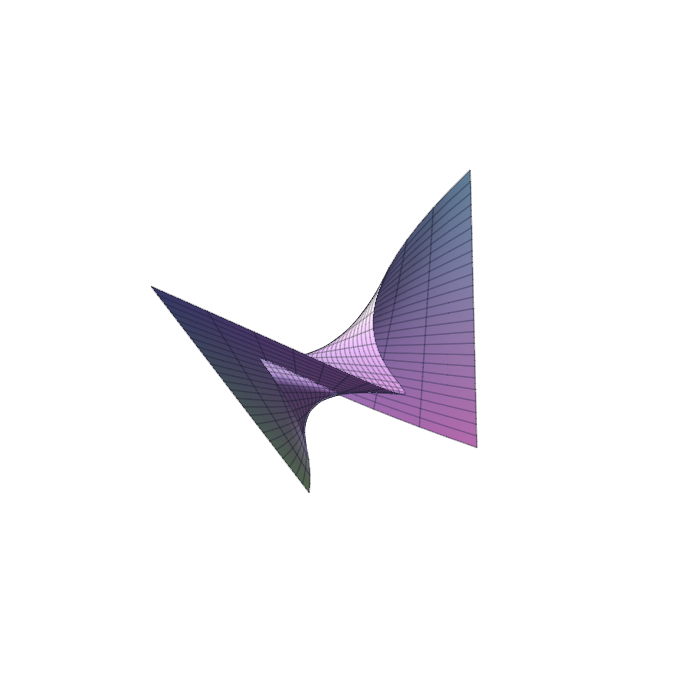}} \\
     $\scalemath{0.7}{\bfx_3(t,s)=\begin{pmatrix}
      t^6 - 6t^4 + t^2 + 2t+s(t^5 - 6t^3 + t)\\[6pt]
      -t^7 + 6t^5 - t^3 + t^2 + t+s(-t^6 + 6t^4 - t^2 + 1)\\[6pt]
      t^3+t+s(t^2+1)
      \end{pmatrix}}$ & \parbox[c]{1em}{\includegraphics[width=1in]{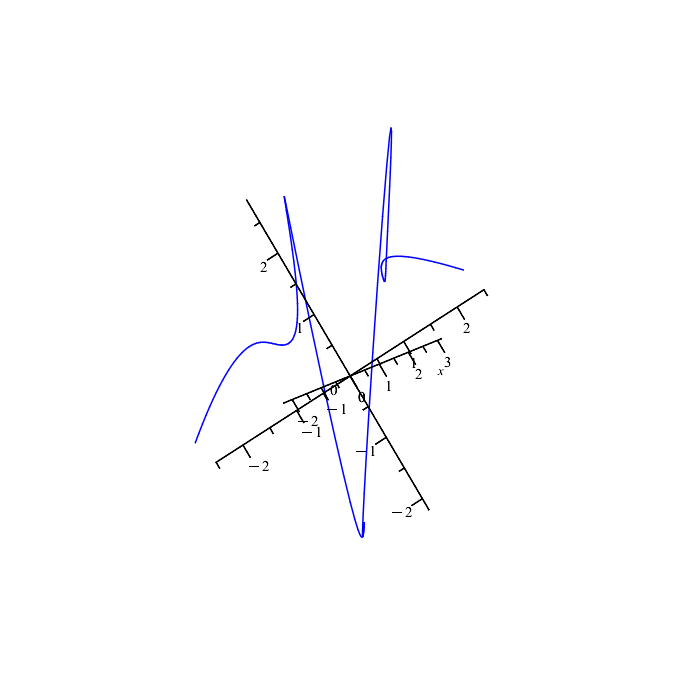}}  &\parbox[c]{1em}{
      \includegraphics[width=1in]{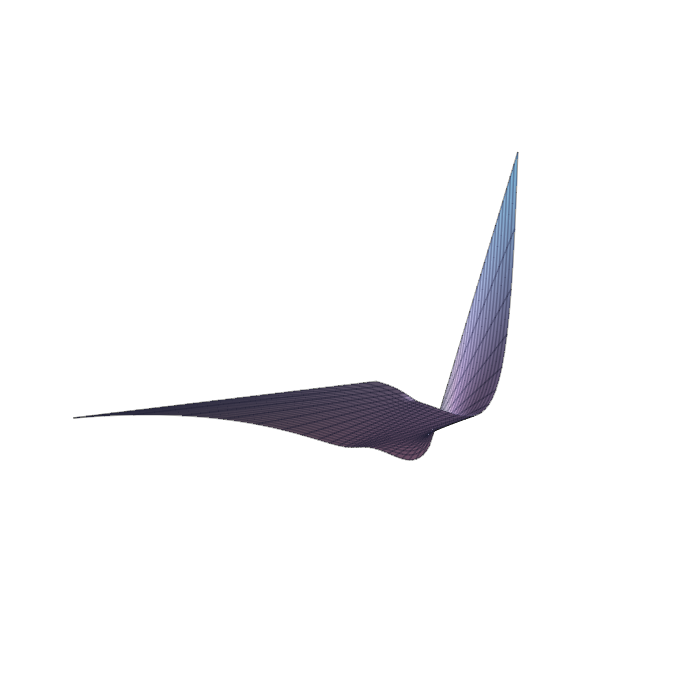}} \\
       $\scalemath{0.7}{\bfx_4(t,s)=\begin{pmatrix}
      \dfrac{t^2}{t^2+1}+st\\[6pt]
      \dfrac{t^4}{t^2+1}+st^3\\[6pt]
      \dfrac{t^5}{t^2+1}+s
      \end{pmatrix}}$ & \parbox[c]{1em}{\includegraphics[width=1in]{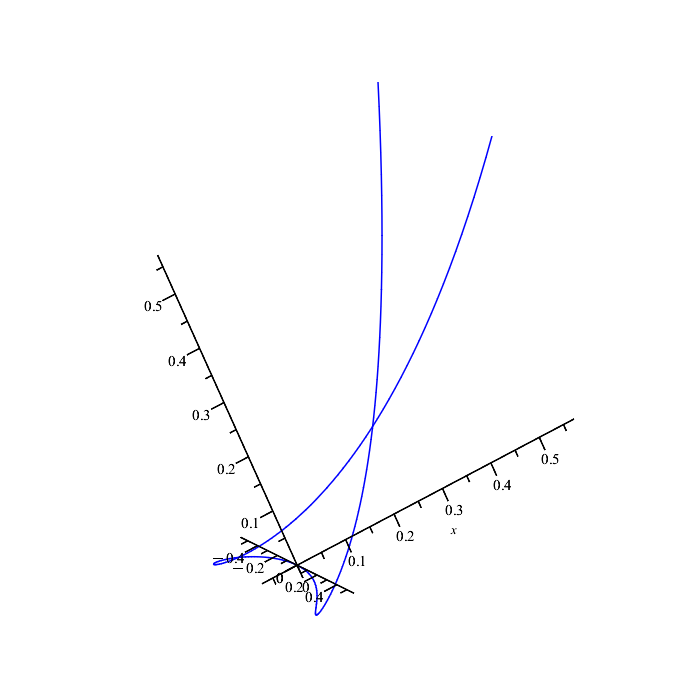}} &\parbox[c]{1em}{
      \includegraphics[width=1in]{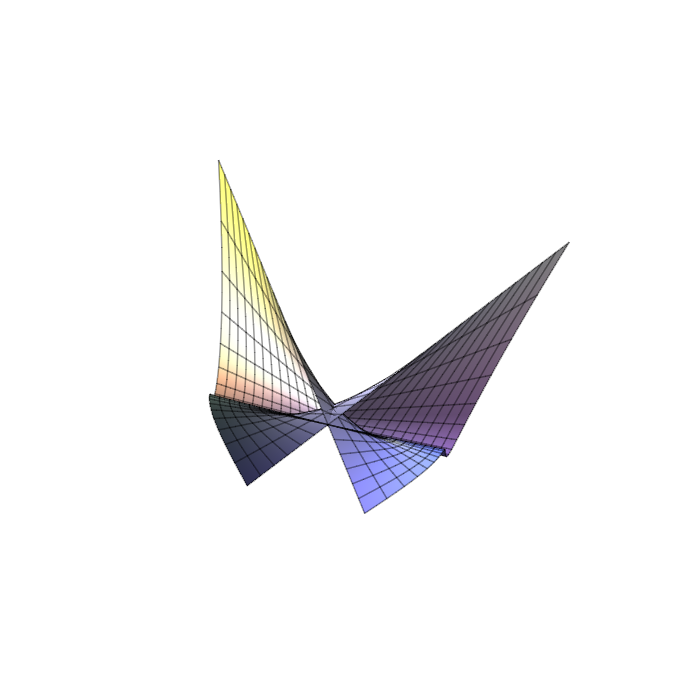}} \\
       $\scalemath{0.7}{\bfx_5(t,s)=\begin{pmatrix}
      4+s(t+1)^2\\[6pt]
      1+s(t+1)\\[6pt]
      t+s
      \end{pmatrix}}$ & \parbox[c]{1em}{\includegraphics[width=1in]{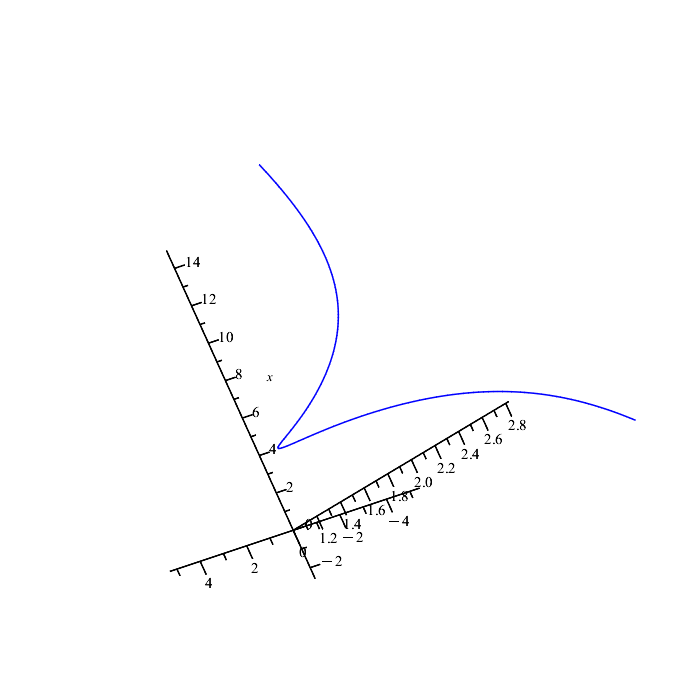}} &\parbox[c]{1em}{
      \includegraphics[width=1in]{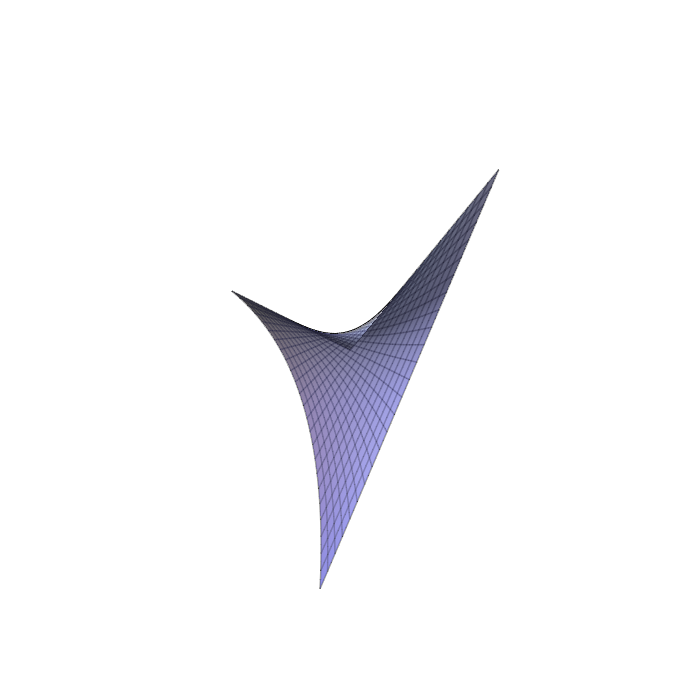}} \\
      $\scalemath{0.7}{\bfx_6(t,s)=\begin{pmatrix}
      \dfrac{t^3}{t^2+1}+s(-t^5+t)\\[6pt]
      \dfrac{t^5}{t^2+1}+3st^7\\[6pt]
      \dfrac{t^7}{t^2+1}2st^3
      \end{pmatrix}}$ &\parbox[c]{1em}{\includegraphics[width=1in]{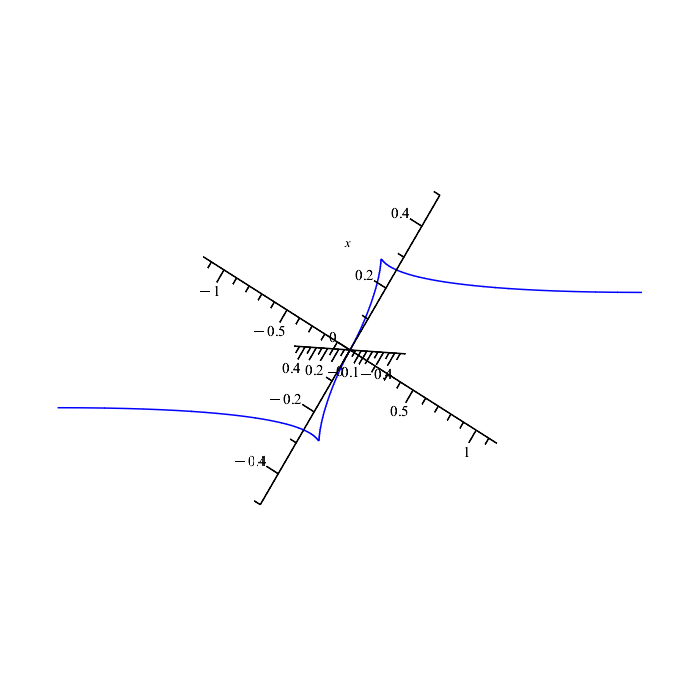}} &\parbox[c]{1em}{
      \includegraphics[width=1in]{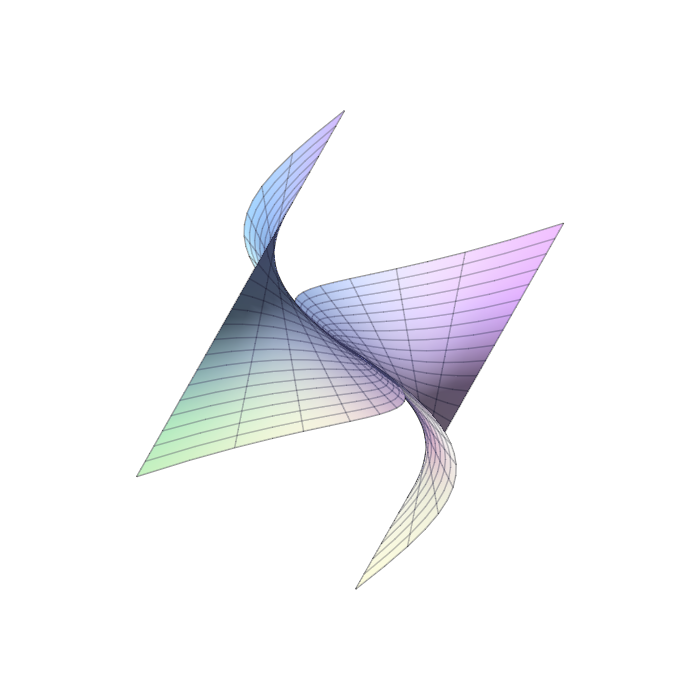}} \\
        $\scalemath{0.7}{\bfx_7(t,s)=\begin{pmatrix}
      t^4 + t^2 + t+s(t^3+t)\\[6pt]
      t^6+t^3+st^5\\[6pt]
      t^5 + t^3 + t^2 + 3t+s(t^4 + t^2 + 3)
      \end{pmatrix}}$ & \parbox[c]{1em}{\includegraphics[width=1in]{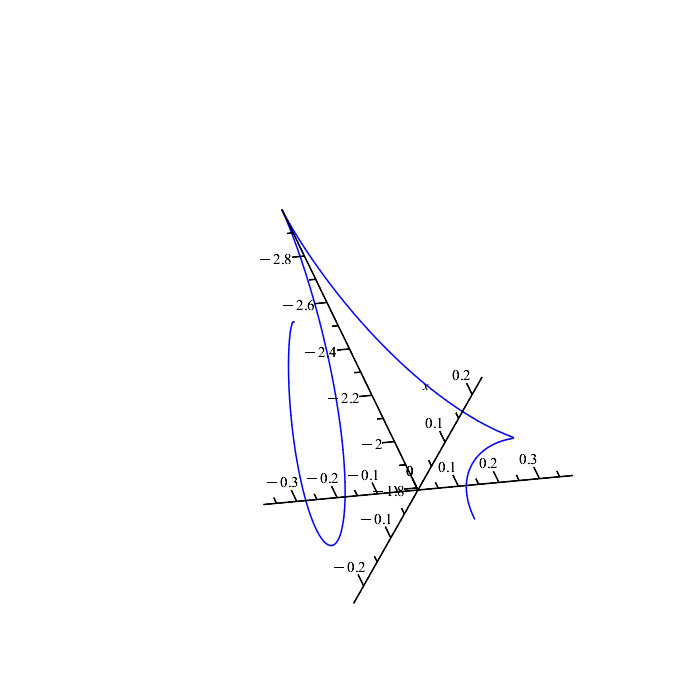}}  &\parbox[c]{1em}{
      \includegraphics[width=1in]{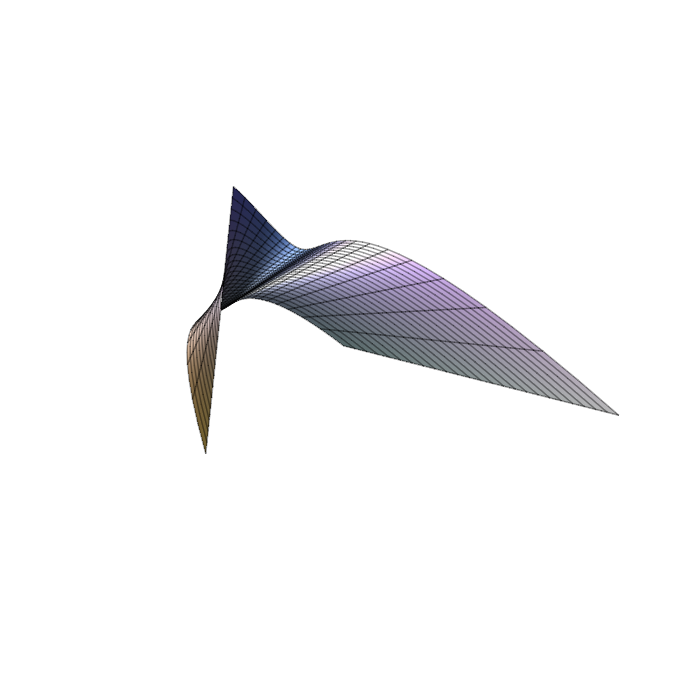}} \\
       $\scalemath{0.7}{\bfx_8(t,s)=\begin{pmatrix}
      -\dfrac{t^{17} - 6t^{15} + 6t^{11} - 6t^7 + 6t^3 - t^2 - t + 1}{t^2+1}+s(-t^6 + 7t^4 - 7t^2 + 1)\\[6pt]
      -\dfrac{2t^{16} - 10t^{14} - 10t^{12} + 2t^{10} + 2t^8 - 10t^6 - 10t^4 + 2t^2 + 1}{t^2+1}+s(2t^5 -12t^3 + 2t)\\[6pt]
      t(t^2+1)^3(t^8+1)+s(t^2+1)^3
      \end{pmatrix}}$ & \parbox[c]{1em}{\includegraphics[width=1in]{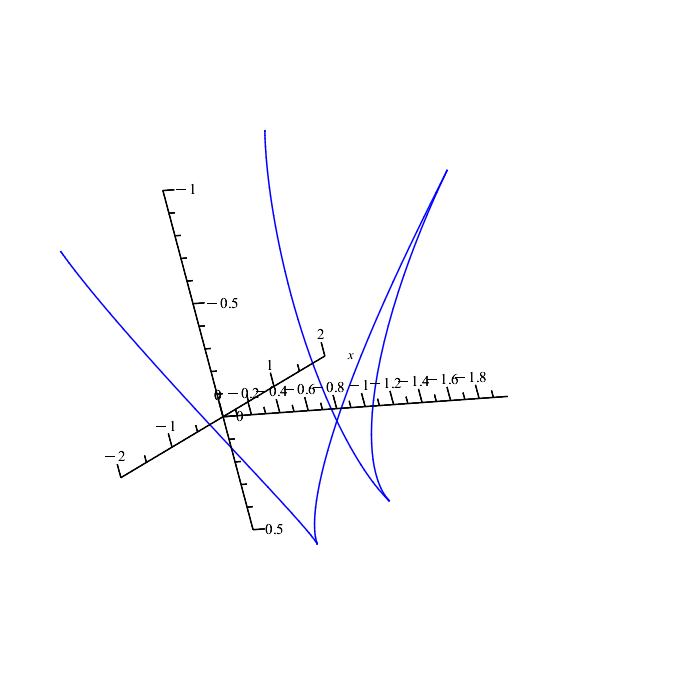}} &\parbox[c]{1em}{
      \includegraphics[width=1in]{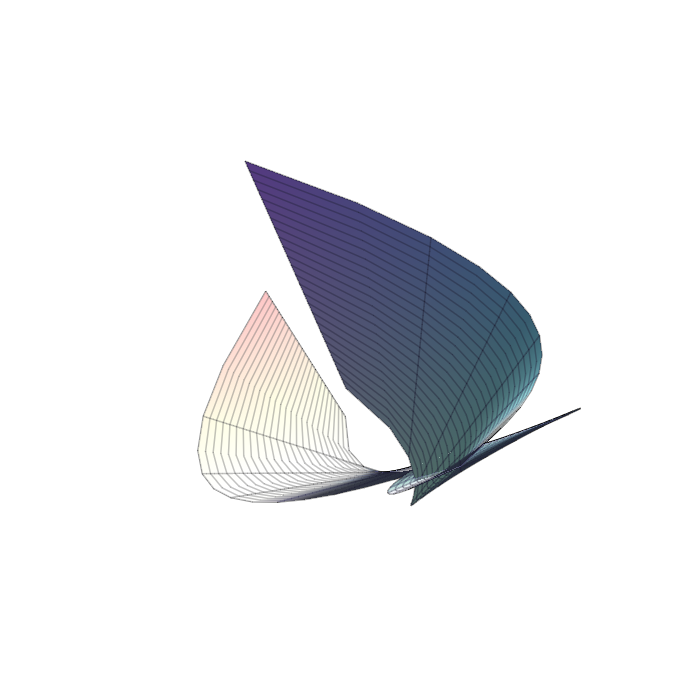}} \\
      \hline
  \end{tabular}
  \caption{Some ruled surfaces and their lines of striction} \label{tab:surf}
\end{table}

\section{Conclusion}\label{conclusion}

While the problem of computing the symmetries of rational curves\textcolor{black}{, and related questions, have received much attention in the literature}, \textcolor{black}{similar problems for rational surfaces still pose several challenges}. In this paper we have provided two algorithms related to this question. The first one is simple, very general, and, at least theoretically, can be applied to a wide variety of parametrizations. However, the computations can be hard, so in practice one requires some kind of advantage in the parametrizations: some cases where the algorithm works well, as we have shown in the experimentation section, are the cases of PN surfaces, toric surfaces and Pl\"ucker conoids. The second algorithm we provide is specific for ruled surfaces, and proceeds by reducing the problem to the case of space curves (which we can solve efficiently) by using the line of striction of the surface. As also shown in the experimentation section, this \textcolor{black}{last} algorithm beats other algorithms that have been considered for the same problem. 

We do not claim that the first algorithm provided in this paper should always be used. Apparently, all the algorithms presented so far to solve the problem considered in this paper have some kind of limitation or drawback. We just intend to contribute to the problem with an algorithm that can be useful in some cases, and perhaps in cases where other algorithms may fail. As for the second algorithm, for ruled surfaces, we dare say that the algorithm is more efficient than other algorithms presented to this day. Nevertheless, for this second algorithm it can certainly happen that the line of striction is a line or a circle, in which case we would have infinitely many symmetries to test. In that case, one should choose an alternative algorithm among the ones reviewed in this paper. 

\textcolor{black}{The ideas presented in this paper can be extended without any effort to compute isometries between rational surfaces, both in the cases of non-ruled and ruled surfaces. As for other generalizations, in order to detect, for instance, affine or projective equivalences between rational surfaces we would need to replace the Gauss and mean curvature by affine or projective invariants}. Finding invariants with the properties that we require, which include a ``good" behavior with respect to reparametrizations, is not a trivial task (see \cite{Gozutok} for a similar question for curves). So we leave this question here as an open problem.

\section*{Use of AI tools declaration}
The authors declare they have not used Artificial Intelligence (AI) tools in the creation of this article.

\section*{Acknowledgments}

Juan Gerardo Alc\'azar and Carlos Hermoso are supported by the grant PID2020-113192GB-I00 (Mathematical Visualization: Foundations, Algorithms and Applications) from the Spanish MICINN. Juan G. Alc\'azar and Carlos Hermoso are also members of the Research Group {\sc asynacs} (Ref. {\sc ccee2011/r34}). Juan Gerardo Alc\'azar and U\u{g}ur G\"oz\"utok are supported by TUBITAK (The Scientific and Technological Research Council of Türkiye) in the scope of 2221-Fellowships for Visiting Scientists and Scientists on Sabbatical Leave.

\section*{Conflict of interest}

The authors declare there is no conflicts of interest.

\end{document}